%% file: main.tex
\pdfoutput=1 
\RequirePackage{amsmath}
\documentclass[envcountsame,envcountsect,runningheads]{llncs}
\input{preamble}
\title{Discrete \& Bayesian Transaction Fee Mechanisms}
\author{%
    Yotam Gafni\inst{1}%
    \orcidID{0000-0002-2144-655X}%
    \and%
    Aviv Yaish\inst{2}%
    \orcidID{0000-0002-7971-2494}%
}
\institute{%
    Weizmann Institute of Science, Rehovot, Israel%
    \iflong,\\\email{yotam.gafni@gmail.com}%
    \\\url{https://www.yotamgafni.com}\fi%
    \and
    The Hebrew University, Jerusalem, Israel%
    \iflong,\\\email{aviv.yaish@mail.huji.ac.il}%
    \\\url{https://www.avivyaish.com}\fi%
}
\begin{document}
\maketitle
\begin{abstract}
    Cryptocurrencies employ auction-esque \emph{\glspl{TFM}} to allocate transactions to blocks, and to determine how much fees miners can collect from transactions.
    Several impossibility results show that \glspl{TFM} that satisfy a standard set of ``good'' properties obtain low revenue, and in certain cases, no revenue at all.
    In this work, we circumvent previous impossibilities by showing that when desired \gls{TFM} properties are reasonably relaxed, simple mechanisms can obtain strictly positive revenue.
    By discretizing fees, we design a \gls{TFM} that satisfies the extended \gls{TFM} desiderata: it is \gls{DSIC}, \gls{MMIC}, \gls{SCP} and obtains asymptotically optimal revenue (i.e., linear in the number of allocated bids), and optimal revenue when considering separable \glspl{TFM}.
    If instead of discretizing fees we relax the \gls{DSIC} and \gls{SCP} properties, we show that Bitcoin's \gls{TFM}, after applying the revelation principle, is \gls{BIC}, \gls{MMIC}, \gls{OCA} proof, and approximately revenue-optimal.
    We reach our results by characterizing the class of multi-item \gls{OCA}-proof mechanisms, which may be of independent interest.%
    \keywords{Auctions \and Mechanism Design \and Incentives}
\end{abstract}

\iflong\let\thefootnote\relax\footnotetext{This work subsumes the discrete and Bayesian myopic analysis of ``\href{https://arxiv.org/pdf/2210.07793v3}{Greedy Transaction Fee Mechanisms for (Non-)myopic Miners}''. That work's remaining results are now in other papers, namely the strictness of \gls{SCP} relative to \gls{OCA}-proofness (see \cite{gafni2024barriers} for an analysis of collusion in \glspl{TFM}), and the non-myopic analysis of \glspl{TFM} (see \cite{gafni2024competitive} for a study of this setting).}\fi

\section{Introduction}
When blockspace is in high demand, users can attach a fee to their transactions, thus prioritizing their transactions' allocation to a block.
The fees that miners can collect from allocated transactions are determined by the blockchain's \glsxtrfull{TFM}.
A challenge faced by \glspl{TFM} is the ability of miners to manipulate the mechanism for their profit at the expense of users, and the ability of users and miners to increase their joint utility by colluding.

Given this challenge, Roughgarden \cite{roughgarden2021transaction} poses an open question: \emph{can we design ``good'' \glspl{TFM} that are resistant to manipulation?}
Previous work cast doubt on the prospects of designing such \glspl{TFM} via a series of impossibility results, showing that their revenue is 0 in certain cases \cite{chung2023foundations,chung2024collusionresilience,gafni2024barriers}, and that even \glspl{TFM} which rely on strong cryptographic primitives suffer from low revenue \cite{shi2023what}.

\subsection{Our Approach}
We show that previous impossibility results do not close the door on good \glspl{TFM}, as they rely on implicit assumptions that we relax, allowing us to prove that \glspl{TFM} which do not rely on sophisticated cryptographic tools can obtain strictly positive revenue.
In total, our work presents a positive outlook for \glspl{TFM}.

\paragraphNoSkip{\gls{TFM} desiderata}
In \cref{sec:Model}, we present the literature's canonical list of desirable \gls{TFM} properties.
Briefly, a good \gls{TFM} should be:
(1) \glsxtrfull{DSIC} for users (\cref{def:Dsic}), meaning that users' fees correspond to their truthful valuation for having their transactions allocated;
(2) \glsxtrfull{MMIC} for miners (\cref{def:Mmic}), meaning that profit-maximizing myopic miners cannot increase profits by omitting user transactions from blocks or allocating ``fake'' miner-created transactions;
and (3) resistant to miner-user collusion, where several works consider a notion called \glsxtrfull{OCA} proofness (\cref{def:Oca}) \cite{roughgarden2021transaction,chung2024collusionresilience,gafni2024barriers}, while others consider a notion named \glsxtrfull{SCP} (\cref{def:Scp}) \cite{chung2023foundations,shi2023what}.

\paragraphNoSkip{Characterizing \gls{OCA}-proofness}
We begin with our main working tool, \cref{thm:OcaChar}, which extends the characterization of single-item \gls{OCA}-proof \glspl{TFM} by Gafni and Yaish \cite{gafni2024barriers} to multi-item mechanisms.
Importantly, we find that if such \glspl{TFM} rely on burning fees (i.e., taking them out of circulation instead of paying them to the miner), then the amount of burnt fees depends only on the number of allocated transactions.
Later, we use this result to reason about the revenue of various \gls{TFM} designs in the discrete and Bayesian settings.

\paragraphNoSkip{Possibility \& impossibility in the discrete case}
We construct a \gls{TFM} in \cref{sec:Discrete} that ticks all the boxes: it is \gls{DSIC}, \gls{MMIC} and \gls{OCA}-proof, thus giving a positive answer to the open question of Roughgarden \cite{roughgarden2021transaction}.
Furthermore, we show in \cref{thm:Gta} that our \gls{TFM} satisfies \gls{SCP} and obtains positive revenue, where the impossibility result of \cite{chung2023foundations} is circumvented by \emph{discretizing} fees, that is, users can only specify fees that correspond to integers.
We follow with \cref{thm:ScpImpliesLowRev}, an impossibility result showing that the revenue of \gls{DSIC} and \gls{SCP} discrete-fee \glspl{TFM} is linear in the number of transactions included in a block, implying the optimality of our \gls{TFM} up to a constant.
We follow with \cref{thm:SeparableScpImpliesLowRev}, a tight upper-bound for the subclass of separable \glspl{TFM} (\cref{def:SeparableRules}).

\paragraphNoSkip{Possibility in the continuous case}
In the traditional setting where fees are continuous, we show in \cref{sec:PABGA} that one can design good \glspl{TFM} by relaxing the \gls{DSIC} and \gls{SCP} requirements, and prove that Bitcoin's \gls{TFM}, also known as the \glsxtrfull{PABGA}, is \glsxtrfull{BIC}, \gls{MMIC}, \gls{OCA}-proof, and obtains good revenue with respect to three benchmarks:
\begin{enumerate}[leftmargin=*,noitemsep]
    \item In \cref{thm:BKPABGA}, we prove a Bulow-Klemperer \cite{bulow1996auctions} type result.
    We show that with any \gls{iid} regular distribution, the \gls{PABGA} with $n$ users and a block size of $\blocksize$ transactions has $\frac{n-\blocksize}{n}$ the revenue of the optimal \gls{TFM} with $\blocksize$ identical items.
    \item We examine special distributions of interest in \cref{clm:UniformDistB,clm:ExponentialDist}, and show that the \gls{PABGA}'s revenue is a significant constant fraction of the total value surplus for uniform users with size-limited blocks, and for exponential users when there are many more transactions than a single block can fit.
    \item Finally, we analyze a large and natural class of \glspl{TFM} we call ``well-reserved'', that behave properly with respect to their reserve prices (that specify the fee below which transactions cannot be allocated, and which may be $0$).
    This class includes Ethereum's \gls{TFM} \glsxtrshort{EIP}1559 \iflong\cite{buterin2019eip}\fi and Myerson's auction \cite{myerson1981optimal}.
    We show in \cref{thm:RevenueOptimalClass} that within this class, the \gls{PABGA} is the revenue optimal \gls{BIC}, \gls{MMIC} and \gls{OCA}-proof \gls{TFM}.
    In \cref{ex:NonConstantDelta}, we highlight designs that are not well-reserved, and outperform the \gls{PABGA} by restricting supply.
\end{enumerate}

All missing proofs are provided in \cref{sec:Proofs}.

\subsection{Related Work}
\label{sec:RelatedWork}

The study of \glspl{TFM} was initiated by Lavi \emph{et al.} \cite{lavi2019redesigning}, who propose the ``monopolistic auction'' of Goldberg \& Hartline \cite{goldberg2001competitive} where all included transactions pay the same fee as an alternative to Bitcoin's \gls{TFM}, and show it can better maintain miner revenue (and thus, Bitcoin's security) when considering large block sizes and low block rewards.
The authors' conjecture that their mechanism is \gls{IC} when the number of bidders approaches infinity was resolved by Yao \cite{yao2018incentive}.
Basu \emph{et al.} \cite{basu2023stablefees} combine the mechanism proposed by \cite{lavi2019redesigning} with the observation of Eyal \emph{et al.} \cite{eyal2016bitcoin}, that splitting fee revenue between miners of consecutive blocks can prevent miner misbehavior, and show that the resulting \gls{TFM} becomes more robust to manipulations when the number of users and miners increases.
Roughgarden \cite{roughgarden2021transaction} crystallizes the observations of Lavi \emph{et al.} \cite{lavi2019redesigning} by formally defining \gls{DSIC}, \gls{MMIC} and \gls{OCA}, and moreover showing that \glsxtrshort{EIP}1559 mostly satisfies these properties.
Chung and Shi \cite{chung2023foundations} suggest a collusion notion called \gls{SCP}, and show that \gls{DSIC} and \gls{SCP} mechanisms cannot obtain more than $0$ revenue, even if blocks are infinite in size.
We circumvent their result in \cref{sec:Discrete} by relying on discrete bids.
Gafni and Yaish \cite{gafni2024barriers}, and Chung \emph{et al.} \cite{chung2024collusionresilience} show that \glspl{TFM} that satisfy the \gls{DSIC}, \gls{MMIC} and \gls{OCA}-proofness properties have 0 revenue.
In \cref{sec:PABGA}, we show that \gls{MMIC} and \gls{OCA}-proof \glspl{TFM} can obtain better revenue and satisfy \gls{BIC}, a relaxation of the \gls{DSIC} property.

\section{Model}
\label{sec:Model}
We now define our model\iflong, with a summary of all notations given in \cref{sec:Glossary}\fi.
Following prior art \cite{lavi2019redesigning,yao2018incentive,basu2023stablefees,roughgarden2021transaction,chung2023foundations,shi2023what,chung2024collusionresilience,gafni2024barriers}, we adopt the perspective of auction theory: miners are auctioneers, users are bidders, transaction fees are bids, all transactions are of equal size, miners and users are rational and myopic (i.e., their utility functions only account for the upcoming block), and each user values having a single transaction allocated (irrespective of considerations like transaction ordering, as discussed in \cref{sec:Discussion}).
We use the above terms interchangeably, e.g., ``auction'' is used to refer to a \gls{TFM}.

\subsection{Transaction fee mechanism}
A \gls{TFM} $\auction$ is described by its \emph{intended allocation rule} $\alloc^\auction$, \emph{payment rule} $\pay^\auction$ and \emph{burning rule} $\burn^\auction$.
To formally define these rules, we follow the prevalent nomenclature of the auction theory literature: given some vector $v$, we denote the vector's $i$-th element by $v_i$, and we use $v_{-i}$ to refer to the vector where the $i$-th index is removed.
\Gls{wlog}, we assume the number of pending transactions waiting for allocation is $n$.

\paragraphNoSkip{Allocation rule}
The miner elected to create the upcoming block has complete control over the contents of the block as long as the number of transactions is lower than the size limit $\blocksize \in \mathbb{N} \cup \left\{ \infty \right\}$.
Thus, we make the distinction between a mechanism's \emph{intended} allocation rule $\alloc^\auction$ as defined by the mechanism's creators, and the allocation rule chosen by the miner, which we denote by $\alloc$.
\begin{definition}[Intended allocation rule]
    \label{def:AllocFunc}
    The intended allocation rule $\alloc^\auction: \mathbb{R}_+^n \rightarrow \left\{0,1\right\}^n$ defines which bids should be allocated by auction $\auction$.
    Given a bid vector $\mathbf{\fee}$, the $i$-th bid $\mathbf{\fee}_i$ is allocated under the intended allocation rule if and only if $\alloc^\auction_i(\fee_i,\mathbf{\fee_{-i}}) = 1$.
    Moreover, we have $\sum_{i=1}^n\alloc^\auction_i(\fee_i,\mathbf{\fee_{-i}}) \le \blocksize$.
\end{definition}
Allocation rules may prohibit the allocation of bids below some \emph{reserve price}.
\begin{definition}[Reserve price]
    \label{def:ReservePrice}
    An auction $\auction$ with intended allocation rule $\mathbf{\alloc^\auction}$ has a \emph{reserve price} $\reserve$ if bids lower than $\reserve$ are never allocated under $\mathbf{\alloc^\auction}$, i.e., for any bid vector $\mathbf{\fee}$ and for any bidder $i$:
    $\fee_{i} < \reserve \Rightarrow \alloc^\auction(\fee_i,\mathbf{\fee_{-i}}) = 0$.
\end{definition}

\paragraphNoSkip{Payment rule}
The payment rule defines the fees that each transaction pays (see \cref{def:PaymentRule}).
Specifically, we consider \gls{EPIR} rules that do not overcharge users (see \cref{def:Epir}).
As the payment is enforced by the mechanism, we do not need to distinguish between the ``intended'' and actual payment.
Note that as discussed in \cref{def:BurningRule}, payments may be partially (or fully) burnt, with the remainder transferred as revenue to the miner.
\begin{definition}[Payment rule]
    \label{def:PaymentRule}
    The payment rule $\mathbf{\pay^\auction}:\mathbb{R}^n\rightarrow \mathbb{R}^n$ in auction $\auction$ receives bids $\fee_1, \ldots, \fee_n$, and maps each bid $\fee_i$ to its payment $\pay^\auction_i\left(\fee_i, \mathbf{\fee_{-i}}\right)$.
\end{definition}
\begin{definition}[\Gls{EPIR}]
    \label{def:Epir}
    An auction $\auction$ with payment rule $\mathbf{\pay^\auction}$ is called \gls{EPIR} if bidders do not pay when unallocated: $\alloc_i(\fee_i,\mathbf{\fee_{-i}}) = 0 \Rightarrow \pay^\auction_i\left(\fee_i, \mathbf{\fee_{-i}}\right) = 0$, and pay at most their bids when allocated: $\alloc_i(\fee_i,\mathbf{\fee_{-i}}) = 1 \Rightarrow \pay^\auction_i\left(\fee_i, \mathbf{\fee_{-i}}\right) \le \fee_i$.
\end{definition}

\paragraphNoSkip{Burning rule}
The ``burn'' rule (see \cref{def:BurningRule}) sets the amount of collected fees which is taken out of circulation, and thus cannot contribute to the miner's revenue.
As before, we consider rules that do not ``over-burn'', in the sense that the amount of fees burnt from a transaction's payment cannot exceed its payment.
Similarly to the payment rule, the burning rule is enforced by the \gls{TFM}.
\begin{definition}[Burning rule]
    \label{def:BurningRule}
    The burning rule $\mathbf{\burn^{\auction}}:\mathbb{R}^n\rightarrow \mathbb{R}^n$ in auction $\auction$ receives bids $\fee_1, \ldots, \fee_n$, and maps each bid $\fee_i$ to its ``burn'' amount $\burn^{\auction}_i\left(\fee_i, \mathbf{\fee_{-i}}\right)$, meaning the part of $\fee_i$'s payment that is ``burnt'' and taken out of circulation.
\end{definition}
\begin{definition}[\Gls{EPBB}]
    \label{def:Epbb}
    We say that an auction $\auction$ with payment rule $\mathbf{\pay^\auction}$ and burn rule $\mathbf{\burn^\auction}$ is \gls{EPBB} if the burnt amount cannot exceed a bidder's payment: $\burn^\auction(\mathbf{\fee}) \leq \pay^\auction(\mathbf{\fee})$.
\end{definition}

\begin{myremark}[Dynamic and static \glspl{TFM}]
    \label{rmk:DynamicTFMs}
    The reserve price of \glsxtrshort{EIP}1559, also known as the base fee, is set \emph{dynamically} for each block, according to the utilized capacity of the preceding block, and burns the reserve price for each allocated transaction.
    For dynamic \glspl{TFM}, it is natural to specify an \emph{update rule} that modifies the auction's various rules over time, e.g., by considering the auction's outcomes for prior blocks (similarly to \glsxtrshort{EIP}1559), or the previously submitted bids (a-la Ferreira \emph{et al.} \cite{ferreira2021dynamic}).
    As our focus is on static (i.e., non-dynamic) \glspl{TFM}, we simplify the discussion and do not specify update rules.
\end{myremark}

\subsection{Actors}

\paragraphNoSkip{Users}
Each user has a private value to have a single transaction accepted to the blockchain.
User valuations are drawn \gls{iid} from a probability distribution with a bounded support.
A \emph{bidding strategy} $s$ is a mapping from a valuation vector to a bid vector.
Given a valuation vector, a bid vector and an allocation, we define in \cref{def:MyopicUserUtility} the standard notion of user utility.

\begin{definition}[User utility]
    \label{def:MyopicUserUtility}
    In a single-round auction $\auction$, where the upcoming miner uses allocation function $\alloc$, the utility of a user with valuation $v_i$ sending a transaction with fee $\fee_i$ is:
    $
        u_i(\fee_i, \mathbf{\fee_{-i}} ; v_i)
        \define
        v_i \cdot \alloc_i(\fee_i,\mathbf{\fee_{-i}})
        -
        \pay^\auction_i\left(\fee_i,\mathbf{\fee_{-i}}\right)
        .
    $
\end{definition}

\paragraphNoSkip{Miner}
Miners can choose an allocation strategy $\alloc$ that does not correspond to the intended allocation rule $\alloc^\auction$, and may also create ``fake bids''.
Given some strategy, we define the miner's corresponding utility in \cref{def:MinerUtility}.
For brevity, the allocation is omitted from our notation.

\begin{definition}[Miner utility]
    \label{def:MinerUtility}
    Given payment rule $\mathbf{\pay^\auction}$, burning rule $\mathbf{\burn^\auction}$, and bids $\mathbf{\fee} \in \mathbb{R}_+^{n'}$ where \gls{wlog} the first $n' \le n$ bids are ``real'' users, the miner's utility is:
    $
        u_{miner} (\mathbf{\fee} ; \mathbf{v})
        \define
        \sum_{i=1}^{n'} \left(\pay(b_i, \mathbf{b}_{-i}) - \burn(b_i, \mathbf{b}_{-i}) \right) - \sum_{i=n+1}^{n'} u_i(\fee_i, \mathbf{\fee_{-i}} ; 0)
        .
    $
\end{definition}

\paragraphNoSkip{Joint utility}
\cref{def:MyopicUserUtility,def:MinerUtility} consider our different actors separately, but one may consider the \emph{joint utility} of the miner and a coalition of users.
\begin{definition}[Joint utility]
    \label{def:JointUtility}
    Let $\auction$ be an auction with payment rule $\mathbf{\pay^\auction}$ and burning rule $\mathbf{\burn^\auction}$.
    The \emph{joint utility} of the miner and a coalition of users $C$ under allocation $\alloc$, is:
    $
        u_{joint}(C, \mathbf{\fee}; \mathbf{v})
        \define
        u_{miner} (\mathbf{\fee})
        +
        \sum_{i \in C} u_i(\fee_i, \mathbf{\fee_{-i}} ; v_i).
    $
\end{definition}

\subsection{The TFM Desiderata}
We now define the canonical desired \gls{TFM} properties considered by the literature, and use them to present two auction designs that we examine in our analysis.

\paragraphNoSkip{User incentive compatibility}
A canonical desired property for auctions is that bidding truthfully is \gls{DSIC} for users: the auctioneer's revenue may decrease if they underbid, while bidders may be unable to pay if they overbid.
\begin{definition}[\Glsxtrfull{DSIC}]
    \label{def:Dsic}
    An auction is \gls{DSIC} if it is always best for the bidder to declare its true value, thus for any bidder $i$ with true value $v_i$, for any bid $\fee_i$, and for any bids $\mathbf{\fee_{-i}}$ by the other bidders:
    $u_i(v_i, \mathbf{\fee_{-i}}) \geq u_i(\fee_i, \mathbf{\fee_{-i}})$.
\end{definition}
Bitcoin's \gls{TFM}, given in \cref{def:Pabga}, is \gls{MMIC} and \gls{OCA}-proof, but not \gls{DSIC} \cite{roughgarden2021transaction}.
\begin{definition}[\Glsxtrfull{PABGA}]
    \label{def:Pabga}

    \noindent {\bf Intended allocation rule:} allocate the $\blocksize$ highest-fee pending transactions.

    \noindent {\bf Payment rule:} each allocated transaction pays its entire fee.

    \noindent {\bf Burning rule:} no fees are burnt.
\end{definition}
In the \gls{PABGA}, a winning bidder may be able to lower its bid and still win, if there is a gap between its true value and the losing bids.
In contrast, the \gls{SPA} is \gls{DSIC} \cite{vickrey1961counterspeculation}, where the auction is defined as allocating a single item to the highest bidder, who pays the second-highest bid.
The \gls{UPGA} (\cref{def:Upga}) extends \Glspl{SPA} to $\blocksize > 1$ items.
\begin{definition}[\Glsxtrfull{UPGA}]
    \label{def:Upga}

    \noindent {\bf Intended allocation rule:} allocate the $\blocksize$ highest-bid pending transactions.

    \noindent {\bf Payment rule:} all allocated bids' payment equals the $\left(\blocksize + 1\right)$-th highest bid.

    \noindent {\bf Burning rule:} no fees are burnt.
\end{definition}

\paragraphNoSkip{Miner incentive compatibility}
The \gls{UPGA} (and thus, the \gls{SPA} on a single item), has an important shortcoming: since all payments equal the $\left( \blocksize + 1 \right)$-th largest bid, the miner may add a ``fake'' bid equal to the $\blocksize$ largest bid to increase its revenue.
To prevent such malfeasance, \glspl{TFM} should be \gls{MMIC}.
\begin{definition}[\Glsxtrfull{MMIC}]
    \label{def:Mmic}
    We say that a \gls{TFM} is \gls{MMIC} if the miner cannot strictly increase its utility by creating bids of its own, or by deviating from the intended allocation rule.
\end{definition}

\paragraphNoSkip{Notions of collusion}
A coalition of users and miners may collude to increase their joint utility.
When colluding, we allow both the miner and users to freely deviate from their corresponding ``honest'' behaviors.
The literature explored two collusion-resistance notions: \gls{OCA}-proofness and \gls{SCP}.
\begin{definition}[\Glsxtrlong{OCA}-proof (\glsxtrshort{OCA}-proof)]
    \label{def:Oca}
    An auction is $c$-\gls{OCA}-proof if a coalition of a miner and up to $c$ bidders cannot increase their joint utility to be higher than that of the intended allocation's winning coalition.
    If an auction is $c$-\gls{OCA}-proof for any $c$, we say it is \gls{OCA}-proof.
\end{definition}
\begin{definition}[\Glsxtrfull{SCP}]
    \label{def:Scp}
    An auction $\auction$ is $c$-\gls{SCP} if a collusion of up to $c$ bidders and a miner cannot increase their joint utility by deviating from the honest protocol.
    If $\auction$ is $c$-\gls{SCP} for any $c$, we say it is \gls{SCP}.
\end{definition}
\begin{myremark}
    Gafni and Yaish \cite{gafni2024barriers} find that \gls{OCA}-proofness is weaker than \gls{SCP}, as the latter implies the former, but not vice-versa.
    Intuitively, a miner-user collusion that breaks \gls{OCA}-proofness allows the colluding coalition to ``steal'' from the \gls{TFM} (e.g., they increase their joint utility above the one that is prescribed by the mechanism's intended allocation), while a collusion that breaks \gls{SCP} increases the colluding coalition's joint utility at the expense of other users.
\end{myremark}

\section{Characterizing \gls{OCA}-proof Auctions}
We now present our main working tool, \cref{thm:OcaChar}, that we use to show revenue-optimality for both the discrete and Bayesian settings.
We do so by extending the characterization of single-item $1$-\gls{OCA}-proof \glspl{TFM} by Gafni and Yaish \cite{gafni2024barriers} (see their Lemma~\href{https://arxiv.org/pdf/2402.08564v1\#lemma.4.4}{4.4}) to multiple items.
Their characterization shows that such auctions are exactly those with non-negative reserve prices, where if the highest bid is above the reserve, it is allocated and the reserve is burnt.

\begin{theorem}[Characterization of \gls{OCA}-proofness]
    \label{thm:OcaChar}
    Any multi-item deterministic \gls{OCA}-proof mechanism $\alloc, \pay, \burn$ is of the following form:
    Its burn rule satisfies the property that the total amount of burnt funds can be described using a monotonically non-decreasing size-based function $\burn^{\text{\tiny cardinal}}:[\blocksize] \rightarrow R_+$, so that if $\sum_{i=1}^n \alloc_i(\fee_i, \mathbf{\fee_{-i}}) = k$, then $\burn^{\text{\tiny cardinal}}(k) \define \sum_{i=1}^n \burn_i(\fee_i, \mathbf{\fee_{-i}})$.
    The allocation rule does not allocate more than $\blocksize'$ bids (where $\blocksize'$ is possibly lower than the externally enforced block size $\blocksize$), and, assuming bids are in decreasing order, it allocates the first $k^*$ bids, for some
    $k^* \in \argmax_{0 \leq k \leq \blocksize'} \left( \left( \sum_{i=1}^k \fee_i \right) - \burn^{\text{\tiny cardinal}}(k) \right).$
\end{theorem}
\begin{proof}
    We begin by characterizing the burn function.
    For a given number of allocated bids $k$, the sum of all burnt fees must be constant, or the bidders and the miner can collude to lower it and increase their joint utility.
    For each possible amount of allocated bids $k$ and the matching total burn amount $\burn_k$, we define the function $\burn^{\text{\tiny cardinal}}(k) = \burn_k$.
    The function must be monotonically non-decreasing, i.e., $\forall k \le k': \burn^{\text{\tiny cardinal}}(k) \le \burn^{\text{\tiny cardinal}}(k')$.
    Otherwise, the miner could collude with $k' > k$ bidders to lower the burn, as $\burn^{\text{\tiny cardinal}}(k') < \burn^{\text{\tiny cardinal}}(k)$.

    Let $\alpha^{feasibility}: N \rightarrow \{0,1\}$ be the size-based ``allocation feasibility function'' that determines whether there is a setting where the mechanism allocates a certain number of bids $k$:
    $
    \alpha^{feasibility}(k) = \begin{cases}
    1 & \exists \mathbf{b} \text{ s.t. } \sum_{i=1}^n \alloc_i(\mathbf{b}) = k \\
    0 & Otherwise.
    \end{cases}
    .
    $
    We proceed by showing that the size-based allocation feasibility function is monotonically non-increasing with the number of bidders, i.e., there is some ``inferred'' block size $\blocksize'$ (possibly infinite, or lower than the externally imposed $\blocksize$), so that it is possible to achieve any number of allocated bids as high as $\blocksize'$, but not higher.
    Assume toward contradiction that there is a bid vector $\mathbf{\fee}$ so that $k'$ bids are allocated, but for some $k < k'$ there is no vector such that exactly $k$ are allocated.
    Consider $k$ bidder valuations all equal to $\burn^{\text{\tiny cardinal}}(k') + 1$.
    Since by assumption at most $k-1$ may be allocated, the bidders' joint utility with the miner cannot exceed $u = (k-1)\cdot (\burn^{\text{\tiny cardinal}}(k') + 1)$.
    The \gls{OCA} where the bidders and the miner bid $\mathbf{\fee}$ (note the miner is required to inject fake bids) results in all $k$ bidders being allocated, by our assumption.
    This collusion gives a joint utility of at least $k\cdot (\burn^{\text{\tiny cardinal}}(k') + 1) - \burn^{\text{\tiny cardinal}}(k') \geq u + 1$.

    Finally, for any number of allocated bids $k$, the allocation rule must allocate the highest $k$ bids, or we could permute the bids to make sure the highest $k$ bidders are allocated, and that would be an \gls{OCA}.
    Notice that $\left(\sum_{i=1}^k \fee_i\right) - \burn^{\text{\tiny cardinal}}(k)$ is the joint utility that the bidders and the miner can achieve if $k$ bids are allocated.
    Since it is possible to allocate any number of bids between $0$ and $\blocksize'$, one can devise an \gls{OCA} that obtains this joint utility.
    So, to be \gls{OCA}-proof, the allocation rule must allocate $k^*$ bids that maximize this expression, i.e., $k^* \in \argmax_{0 \leq k \leq \blocksize'} \left( \left(\sum_{i=1}^k \fee_i\right) - \burn^{\text{\tiny cardinal}}(k) \right)$.
\end{proof}

\section{Optimal \gls{SCP} \glspl{TFM} via Discrete Bids}
\label{sec:Discrete}
We present the first \gls{TFM} that satisfies the \gls{TFM} desiderata and obtains positive revenue, in a setting where bids and valuations are discrete.
Although prior work shows that good \glspl{TFM} which satisfy partial or relaxed versions of the desiderata necessarily produce $0$ revenue even when blocks are infinite in size \cite{chung2023foundations}, the corresponding proofs implicitly assume that bids are continuous, as is the standard in auction theory.
However, cryptocurrencies specify fees in some minimal denomination, e.g., \emph{sat} is Bitcoin's smallest unit, and \emph{wei} is Ethereum's.

\subsection{The Good Toy Auction}
In \cref{def:Gta} we present our \gls{TFM}, the \emph{\gls{GTA}}.
Following previous work \cite{goldberg2001competitive,rizun2015transaction,lavi2019redesigning,chung2023foundations}, we first consider the case where blocks are not size-limited.
Later, we show that this is a necessary condition to satisfy all desired \gls{TFM} properties (\cref{thm:ScpImpliesLowRev}).
Despite the abundance of this setting, it is non-trivial: when block-rewards are low (e.g., Bitcoin's rewards converge to 0\iflong \cite{nakamoto2008bitcoin}\fi), fees are an important revenue source, implying that fee revenue should be maintained even if blocks are allowed to be large to preserve the security of the system \cite{lavi2019redesigning}.
In \cref{rmk:GtaFinite} we explain how to modify the \gls{GTA} for finite blocks.
\begin{definition}[\Glsxtrfull{GTA}]
    \label{def:Gta}

    \noindent {\bf Bid space:} bids are discrete, for any bid $\fee_i$: $\fee_i \in \{0\}\cup\mathbb{N}$.

    \noindent {\bf Intended allocation rule:} allocate all pending transactions that bid $1$ or more.

    \noindent {\bf Payment rule:} the payment of an allocated transaction is equal to $1$.

    \noindent {\bf Burning rule:} no fees are burnt.
\end{definition}
The intended allocation of \cref{def:Gta} implicitly assumes that blocks are not size-limited, which is reasonable for settings in which the demand is low, implying that even if there are some capacity constraints, they are usually not reached.
In \cref{rmk:GtaFinite}, we briefly outline how to modify the \gls{GTA} so that it maintains its good properties, even when considering finite blocks.
\begin{myremark}[\gls{GTA} for finite blocks]
    \label{rmk:GtaFinite}
    The \gls{GTA} can be adapted to the case where blocks are size-limited.
    This adaptation preserves most of its good properties, but necessitates relaxing the strict \gls{SCP} requirement to a ``usually'' \gls{SCP} requirement, a-la Roughgarden's usually \gls{OCA} requirement \cite{roughgarden2021transaction}:
    To be usually \gls{SCP}, a \gls{TFM} is required to only be \gls{SCP} when its reserve price $\reserve$ is high enough, such that the amount of available transactions paying fees in excess of $\reserve$ is lower than the block size limit.
    This implies the \gls{TFM} satisfies the \gls{SCP} property when the reserve price is set correctly.
    To adapt the \gls{GTA} accordingly, one can use a dynamic reserve price that is burnt, similarly to \glsxtrshort{EIP}1559 (see \cref{rmk:DynamicTFMs}).
\end{myremark}

In \cref{thm:Gta} we show that the \gls{GTA} satisfies the \gls{TFM} desiderata.
\begin{restatable}{theorem}{thmGta}
    \label{thm:Gta}
    \gls{GTA} is \gls{DSIC}, \gls{MMIC}, \gls{OCA}, \gls{SCP}, and has positive revenue.
\end{restatable}
We provide the corresponding proof in \cref{sec:Proofs}.
Intuitively, \gls{DSIC} is satisfied as the \gls{TFM}'s reserve price is $1$, which is the lowest value allowed, implying that users cannot profit from lower bids (i.e., 0) as they necessarily lead to lower utility. Furthermore, the miner enjoys strictly positive revenue when there is at least one user with a positive valuation, as the mechanism's intended allocation ensures that all users that can contribute to the revenue are allocated, and moreover, the discrete setting and the lack of burning imply that user payments are transferred to the miner, and cannot be possibly split among agents.
Finally, the \gls{MMIC} property is satisfied according to \cref{fct:SeparableMaximizingMmic}, as the mechanism's intended allocation is revenue maximizing, and the mechanism's payment and burning rules are separable (per \cref{def:SeparableRules}).
A result equivalent to \cref{fct:SeparableMaximizingMmic} was proven as Theorem~\href{https://arxiv.org/pdf/2106.01340v3\#theorem.5.2}{5.2} by Roughgarden \cite{roughgarden2021transaction}, we highlight several differences between the statements of the results in \cref{rmk:SeparableMaximizingMmic}.
\begin{definition}[Separable rules]
    \label{def:SeparableRules}
    A rule $f$ is \emph{separable} if the $i$-th bid's output does not depend on $\mathbf{\fee_{-i}}$:
    $\forall \mathbf{\fee_{-i}}, \mathbf{\fee_{-i}'}: f_i\left(\fee_i, \mathbf{\fee_{-i}}\right) = f_i\left(\fee_i, \mathbf{\fee_{-i}'}\right)$.
\end{definition}
\begin{restatable}[Separability + revenue maximization $\Rightarrow$ \gls{MMIC} \cite{roughgarden2021transaction}]{theorem}{SeparableMaximizingMmic}
    \label{fct:SeparableMaximizingMmic}
    An auction $\auction$ with separable payment and burning rules $\mathbf{\pay^\auction}, \mathbf{\burn^\auction}$, and an intended allocation rule $\mathbf{\alloc^\auction}$ which is revenue maximizing under $\mathbf{\pay^\auction}$, is \gls{MMIC}.
\end{restatable}

\subsection{Revenue Bounds}
In \cref{thm:ScpImpliesLowRev} we show that even when considering a minimal set of properties that one may want in an auction, adding the requirement that the \gls{DSIC} and 1-\gls{SCP} properties should be satisfied restricts the revenue of discrete-bid auctions to be asymptotically linear in the number of allocated bidders, thus implying the approximate optimality of the \gls{GTA} within this class.
In our case, we consider mechanisms that do not overcharge users (i.e., that are \gls{EPIR} and \gls{EPBB}).
\begin{restatable}{theorem}{thmScpImpliesLowRev}
    \label{thm:ScpImpliesLowRev}
    Any \gls{EPIR}, \gls{EPBB}, \gls{DSIC} and 1-\gls{SCP} discrete bid \gls{TFM} must have infinite blocksize and can obtain a revenue of at most $O(k)$, where $k$ is the number of allocated bids.
\end{restatable}
A tight bound can be reached by furthermore requiring separability.
\begin{restatable}{theorem}{thmSeparableScpImpliesLowRev}
    \label{thm:SeparableScpImpliesLowRev}
    Any \gls{EPIR}, \gls{EPBB}, \gls{DSIC}, 1-\gls{SCP} and separable discrete bid \gls{TFM} obtains a revenue of at most $k$, where $k$ is the number of allocated bids.
\end{restatable}
To give an intuitive overview of our proof for \cref{thm:SeparableScpImpliesLowRev}, it proceeds by contradiction, and is separated into two cases.
In the first, we show that if a bidder can change its bid to increase its utility, one can contradict the \gls{DSIC} property.
In the second, we show that as a bidder cannot ``unilaterally'' profitably deviate, that any manipulation that increases its utility necessitates colluding with the miner, thus contradicting the \gls{SCP} property.

Although the \gls{GTA} and its finite-blocksize version have certain good properties, \cref{thm:ScpImpliesLowRev,thm:SeparableScpImpliesLowRev} show that their per-block revenue is low and linear in the number of allocated transactions.
In \cref{sec:PABGA}, we show that other solutions can obtain better revenue and still satisfy most of the \gls{TFM} desiderata.

\section{Optimal \gls{OCA}-proof \glspl{TFM} via \gls{BIC}}
\label{sec:PABGA}
We now examine the ``traditional'' setting, where bids are continuous, and blocks are finite $\blocksize \in \mathbb{N}$.
We relax the \gls{DSIC} and 1-\gls{SCP} properties, as prior art found that satisfying both in this setting leads to low revenue \cite{chung2023foundations}.

\subsection{The Shading Auction}
\label{sec:ShadingAuction}
The \gls{PABGA} satisfies \gls{MMIC} and is \gls{OCA}-proof, yet it is not \gls{DSIC} (see \cref{sec:Model}).
However, \gls{DSIC} is a strong requirement that fits a full information setting, while blockchain users may not be fully informed as to other users' transactions: although they may be able to observe some pending transactions, they cannot observe future transactions or those that were sent using private channels.
Instead, our setting can be modeled as a game of incomplete information, where it is reasonable to require a relaxed property called \gls{BIC}, which relies on a more natural equilibrium concept called \gls{BNE}.
\begin{definition}[\Glsxtrfull{BNE}]
    \label{def:BNE}
    An auction has a \glsxtrlong{BNE} if for any bidder $i$ there is a bidding strategy $s_i:V_i \rightarrow R$ such that for any $\fee_i$:
    $
        \expect_{\mathbf{v_{-i}}}{u_i(s_i(v_i), \mathbf{s_{-i}(v_{-i})} ; v_i)}
        \geq
        \expect_{\mathbf{v_{-i}}}{u_i(\fee_i, \mathbf{s_{-i}(v_{-i})} ; v_i)}
        .
    $
\end{definition}
\begin{definition}[\Glsxtrfull{BIC}]
    We say an auction is \glsxtrlong{BIC} if for any bidder $i$, bidding truthfully is a \gls{BNE}.
    Equivalently, for all $i$: $s_i(v_i) = v_i$.
\end{definition}

In particular, we are interested in the \emph{efficiency} admitted by a mechanism in equilibrium, where our \cref{def:Efficiency} follows Hartline \emph{et al.} \cite{chawla2013auctions}\iflong, and Dasgupta and Maskin \cite{dasgupta2000efficient}\fi.
\begin{definition}[Auction efficiency]
    \label{def:Efficiency}
    An equilibrium of an auction is called \emph{efficient} if it maximizes the sum of valuations of the allocated bidders.
\end{definition}

Given \cref{sec:Model}, namely, the assumption of valuations drawn \gls{iid} from a distribution with bounded support, Chawla and Hartline \cite{chawla2013auctions} show that the \gls{PABGA} is has a unique and efficient \gls{BNE}.
Thus, by Myerson's revelation principle (stated in \cref{fct:RevelationPrinciple}), the \gls{PABGA} has a \gls{BIC} representation, i.e., one can design a \gls{BIC} auction that always obtains the same results as the \gls{PABGA}.

\begin{lemma}[Revelation principle \cite{myerson1981optimal}]
    \label{fct:RevelationPrinciple}
    Any feasible mechanism $M$ with a \gls{BNE} has an equivalent \gls{BIC} mechanism which always produces the same outcomes as $M$.
    In this equivalent mechanism, bidders reveal their true valuations, and the auctioneer applies their equilibrium strategies on their behalf.
    The resulting mechanism is called the \emph{direct revelation} of $M$.
\end{lemma}
\begin{myremark}[Direct revelation for \glspl{TFM}]
    \label{rmk:DirectRevelation}
    One caveat to applying the revelation principle in our setting is that when we move the burden of executing the equilibrium strategies of the bidders to the miner, and given that we generally consider a strategic miner (as encompassed by the notion of \gls{MMIC}), a direct revelation mechanism could be subject to manipulation that is averted by the decentralized nature of a \gls{BNE}.
    In particular, for the first-price auction, where the equilibrium depends on the number of bidders $n$, and the revenue is monotonically increasing with $n$, the miner may introduce fake bidders and thus justify playing higher bids on behalf of the bidders.
    To ensure that the revelation principle applies even when accounting for miner malfeasance, one may consider using various tools such as \gls{MPC} \cite{chung2024collusionresilience}, where an auction's rules may be enforced by the mechanism, while so-called shared \glspl{mempool} can guarantee that all agents have the same view on the currently pending transactions \cite{danezis2022narwhal,gai2023scaling}.
    Importantly, the \gls{BNE} \emph{can} be played in a decentralized fashion, and so in this sense considering a \emph{BIC} mechanism is analogous to considering whether the bidders' \emph{BNE} is \gls{MMIC} and \gls{OCA}.
    In other words, given that bidders arrive at a certain \emph{BNE}, the miner cannot ex-post manipulate the allocation to increase its revenue, and the miner and bidders cannot coordinate on an alternative strategy that increases their joint utility.
    The \emph{BIC} notion can thus be considered more as a convenient technical tool to simplify the discussion.
\end{myremark}

When applying the revelation principle to the \gls{PABGA}, the resulting mechanism (given in \cref{def:ShadingAuction}) is called the \emph{shading auction}, due to bidders' \gls{BNE} strategies where they may \emph{shade} (i.e., lower) bids relative to their valuations.
\begin{definition}[Shading auction]
    \label{def:ShadingAuction}

    \noindent {\bf Intended allocation rule:} allocate the $\blocksize$ highest-bid transactions.

    \noindent {\bf Payment rule:} each accepted transaction pays its equilibrium bid under the \gls{PABGA}, i.e., if the \gls{BNE} bidding strategy is $s$, then bidder $i$ pays $s_i(v_i)$.

    \noindent {\bf Burning rule:} no fees are burnt.
\end{definition}
In \cref{lem:ShadingAuction}, we show that the shading auction obtains good properties.
\begin{restatable}{lemma}{lemShadingAuction}
    \label{lem:ShadingAuction}
    The shading auction is \gls{BIC}, \gls{MMIC}, \gls{OCA}-proof and efficient.
\end{restatable}

\subsection{Revenue Guarantees}
We now prove a series of results showing that the \gls{PABGA} (and its direct revelation mechanism, the shading auction) has good revenue guarantees.
We start with \cref{thm:BKPABGA}, which provides guarantees for the case where user valuations are drawn from Myerson regular distributions \cite{myerson1981optimal}, as defined in \cref{def:MyersonRegular}.
\begin{definition}[Myerson regular distribution \cite{myerson1981optimal}]
    \label{def:MyersonRegular}
    Let $F,f$ be the \gls{CDF} and \gls{PDF} of a distribution, respectively.
    The distribution is said to be \emph{Myerson regular} if its virtual value function, defined as
    $
        w(v) \define v - \frac{1 - F(v)}{f(v)}
    $, is non-decreasing in $v$.
\end{definition}
\begin{restatable}{theorem}{BKPABGA}
    \label{thm:BKPABGA}
    With Myerson regular \gls{iid} distributions, the \gls{PABGA} with $n$ bidders and $\blocksize$ items has at least $\frac{n-\blocksize}{n}$ of the revenue of the optimal auction.
\end{restatable}
Our above benchmark is the revenue-optimal auction, but this does not necessarily guarantee how much of the value surplus (see \cref{def:ValueSurplus}) is handed to the miner.
I.e., with $\blocksize$ items and individually rational bidders, one can imagine that a miner somehow extracts all surplus to be its revenue.
Even so, for the uniform and exponential distributions, we can guarantee close approximations.
\begin{definition}[Value surplus]
    \label{def:ValueSurplus}
    The value surplus, given $n$ bidders with valuations $v_1, \ldots, v_n$ and $\blocksize$ items, is $V \define \sum_{i=1}^{\blocksize} v_i$.
\end{definition}

For the uniform distribution, we prove guarantees both in the sense of ratio of expectations and expectation of ratios (see the discussion of Ezra \emph{et al.} \cite{ezra2024prophet}).
Let $UNI([0,1])$ be the uniform distribution on the interval $[0,1]$.
\begin{restatable}{myclaim}{uniformDist}
    With uniform \gls{iid} distributions,
    $
        \frac{\expect_{v_j \sim UNI([0,1])}{Rev^{PABGA}}}{\expect_{v_j \sim UNI([0,1])}{V}}
        =
        \frac{n-\blocksize}{n+1-\frac{\blocksize + 1}{2}}
        .
    $
\end{restatable}

\begin{restatable}{myclaim}{UniformDistB}
    \label{clm:UniformDistB}
    With uniform \gls{iid} distributions and $\blocksize \leq 10$,
    $$
        \expect_{v_j \sim {UNI([0,1])}}{\frac{Rev^{PABGA}}{V}}
        \geq
        \frac{n-\blocksize}{n}
        .
    $$
\end{restatable}

\begin{myremark}
    We conjecture that \cref{clm:UniformDistB} generally holds for any $\blocksize$.
    However, proving it for an arbitrary block size requires proving a general non-holonomic identity, which is outside the scope of our work.
\end{myremark}

Let $EXP(\zeta)$ be the exponential distribution with parameter $\zeta$.
\begin{restatable}{myclaim}{ExponentialDist}
    \label{clm:ExponentialDist}
    With exponential \gls{iid} distributions, for some constant $c$, $\blocksize$ items and $n \in \Omega(\blocksize^c)$ bidders,
    $
        \lim_{\blocksize \rightarrow \infty} \frac{\expect_{v_j \sim EXP(\zeta)}{Rev^{PABGA}}}{\expect_{v_j \sim EXP(\zeta)}{V}} \geq \frac{c-1}{c}
        .
    $
\end{restatable}
Intuitively, the claim is interesting for $c > 1$, as $n \in \Omega(\blocksize^c)$ implies that there are more bidders than there are blocks.

\subsection{Revenue-optimality Among the Class of Well-reserved \glspl{TFM}}
\label{sec:WellReserved}

We complement the above results by proving an even stronger result, using a benchmark that is more suitable for \glspl{TFM} than for general auctions, as not every auction is implementable with good \gls{TFM} properties, as discussed in \cite{lavi2019redesigning,roughgarden2021transaction,gafni2024barriers}. We show that the \gls{PABGA} is revenue-optimal revenue among a rich class of mechanisms.
To support the analysis of this class, we introduce definitions that allow reasoning about a single round auction's revenue.
\begin{definition}[Auction payment]
    The total fees paid in a single round auction $\auction$ with payment rule $\mathbf{p}$, given $n$ bidders with valuations $v_1 \geq \ldots \geq v_n$, is $Pay^{\auction} \define \sum_{i=1}^{n} \pay_i(s(v_i))$, where $s$ is the \gls{BNE} strategy of a bidder under $\auction$.
\end{definition}
\begin{definition}[Auction revenue]
    The revenue of a single round auction $\auction$ with payment rule $\mathbf{p}$ and burn rule with payment rule $\mathbf{\burn}$, given $n$ bidders with valuations $v_1 \geq \ldots \geq v_n$, is $Rev^{\auction} \define \sum_{i=1}^{n} \left(\pay_i(s(v_i)) - \burn_i(s(v_i)\right)$, where $s$ is the \gls{BNE} strategy of a bidder under $\auction$.
\end{definition}
Thus, the revenue of the \gls{PABGA} equals $Rev^{PABGA} = \sum_{i=1}^{\blocksize} s(v_i)$, where $s$ is the \gls{BNE} strategy of bidders.
For the \gls{UPGA}, it is $Rev^{UPGA} = \blocksize \cdot v_{\blocksize + 1}$.

Finally, the class we consider is given in \cref{def:MechanismClass}.
Intuitively, the class includes all \gls{EPIR}, \gls{EPBB} mechanisms with a reserve price that ensure that any agent that does not bid higher than the reserve is not included in the block, and any agent that does is included unless the block is already full.
\begin{definition}
    [Class of Well-reserved \glspl{TFM} $\mathcal{M}$]
    \label{def:MechanismClass}
    $\mathcal{M}$ includes all \gls{EPIR}, \gls{EPBB} mechanisms with an allocation rule $\alloc$ such that for some reserve price $\reserve$, and for every bid $\fee_i$ and a vector of other bids $\mathbf{\fee_{-i}}$, if $\fee_i < \reserve$ then $\alloc_i(\fee_i, \mathbf{\fee_{-i}}) = 0$, and if $\fee_i \geq \reserve$, it cannot be that both $\alloc_i(\fee_i, \mathbf{\fee_{-i}}) = 0$ and $\sum_{i=1}^n \alloc_i(\fee_i, \mathbf{\fee_{-i}}) < \blocksize$.
\end{definition}

Prominent members of this class include the \gls{PABGA} with $\reserve = 0$, the \gls{GTA} with $\reserve = 1$, as well as \glsxtrshort{EIP}1559.
To our knowledge, these are the only mechanisms currently known to be \gls{MMIC} and \gls{OCA}-proof.
It also includes the Myerson auction \cite{myerson1981optimal}, a natural candidate for the optimal revenue mechanism.
However, \cref{clm:BurnReserve} implies that the Myerson auction is not \gls{OCA}-proof because while it sets some reserve price, it does not burn fees.

\begin{restatable}{lemma}{WellReservedOCA}
    \label{lem:wellReservedOCA}
    An \gls{OCA}-proof well-reserved \gls{TFM} allocates the top $\blocksize$ bids that are higher than the reserve-price $\reserve \geq 0$, and burns $\reserve$ for each allocated bid.
\end{restatable}

We next use revenue equivalence, which in our setting is actually a payment equivalence when accounting for burning.
This technical tool does not account for miner manipulations (e.g., introducing fake bids or not following the intended allocation), but rather takes the traditional auction theory perspective.
\begin{restatable}{myclaim}{clmPaymentEquivalence}
    \label{clm:PaymentEquivalence}
    Any two \gls{BIC} \glspl{TFM} that have the same allocation rule, have the same expected payments.
\end{restatable}
\begin{proof}
    This is obtained from the ``revenue equivalence'' theorem \cite{nisan2007algorithmic}, which gives conditions under which all mechanisms with the same allocation rule have the same user payments, in expectation.
    In the blockchain setting, this result can be cast as user payment equivalence, but not necessarily revenue equivalence for miners, due to potential burning.
    The technical conditions of the revenue equivalence theorem require that:
    (1) user types are path-connected; indeed, we assume they are symmetric, which satisfies this property.
    (2) There is a type that has the same payment in both mechanisms; since we discuss \gls{EPIR} mechanisms, the user with value $0$ for a transaction has payment $0$ across all mechanisms.
\end{proof}

We thus characterize the revenue of well-reserved \gls{OCA}-proof auctions:
\begin{restatable}{lemma}{WellReservedOCARevenue}
    \label{clm:BurnReserve}
    Any \gls{BIC}, \gls{OCA}-proof well-reserved \gls{TFM} has the same expected revenue as the \gls{UPGA} with a reserve of $\reserve$ for some $\reserve \geq 0$.
\end{restatable}
\begin{proof}
    This follows from the fact that the allocation rule of \cref{lem:wellReservedOCA} is exactly that of the \gls{UPGA} with a reserve of $\reserve$ for some $\reserve\geq 0$.
    Thus, the payments of any \gls{BIC} mechanism must be the same as \gls{UPGA}$_{\reserve}$. Moreover, as we define the burn rules to behave exactly the same, the expected burn is the same as well.
    Since both payments and burn behave the same, so does revenue.
\end{proof}

We conclude with \cref{thm:RevenueOptimalClass}.
\begin{restatable}{theorem}{thmRevenueOptimalClass}
    \label{thm:RevenueOptimalClass}
    The \gls{PABGA} is the revenue-optimal well-reserved \gls{TFM}.
\end{restatable}
\begin{proof}
    By \cref{clm:PaymentEquivalence}, the \gls{PABGA} with no reserve price is revenue equivalent to the \gls{UPGA} with no reserve, as both have the same allocation and no fee burning:
    $
        \expect_{\mathbf{\fee}}{Rev^{PABGA_0}}
        =
        \expect_{\mathbf{\fee}}{Rev^{UPGA_0}}
        .
    $
    As no fees are burnt, we can write:
    $
        \expect_{\mathbf{\fee}}{Rev^{UPGA_0}}
        =
        \expect_{\mathbf{\fee}}{Pay^{UPGA_0}}
        .
    $
    By using conditional expectation:
    $
        \expect_{\mathbf{\fee}}{Pay^{UPGA_0}}
        \geq
        Pr_b[\bidsMoreReserve > \blocksize] \cdot \expect_{\mathbf{\fee}}{Pay^{UPGA_0} | \bidsMoreReserve > \blocksize}
        .
    $
    With an overfull block, the payments of $UPGA_0$ and $UPGA_\reserve$ are equal, thus:
    $
        Pr_b[\bidsMoreReserve > \blocksize] \expect_{\mathbf{\fee}}{Pay^{UPGA_0} | \bidsMoreReserve > \blocksize}
        =
        Pr_b[\bidsMoreReserve > \blocksize] \expect_{\mathbf{\fee}}{Pay^{UPGA_{\reserve}} | \bidsMoreReserve > \blocksize}
        .
    $
    As burning is non-negative, one can write:
    \begin{multline*}
        Pr_b[\bidsMoreReserve > \blocksize] \expect_{\mathbf{\fee}}{Pay^{UPGA_{\reserve}} | \bidsMoreReserve > \blocksize}
        \geq
        \\
        Pr_b[\bidsMoreReserve > \blocksize] \expect_{\mathbf{\fee}}{Pay^{UPGA_{\reserve}} - \hspace{-0.5cm}\sum_{\fee_i \in \alloc^\auction(\fee_i, \mathbf{\fee_{-i}})} \hspace{-0.5cm} \burn^\auction_i(\fee_i, \mathbf{\fee_{-i}}) | \bidsMoreReserve > \blocksize}
    \end{multline*}
    But, this is exactly $\expect_{\mathbf{\fee}}{Pay^{UPGA_r}}$ when $\reserve$ is burnt from each allocated bid, since by conditional probability:
    \begin{align*}
        \expect_{\mathbf{\fee}}{Pay^{UPGA_0}}
        & =
        Pr_b[\bidsMoreReserve > \blocksize] \expect_{\mathbf{\fee}}{Pay^{UPGA_{\reserve}} - \hspace{-0.5cm}\sum_{\fee_i \in \alloc^\auction(\fee_i, \mathbf{\fee_{-i}})} \hspace{-0.5cm} \burn^\auction_i(\fee_i, \mathbf{\fee_{-i}}) | \bidsMoreReserve > \blocksize}
        \\
        & +
        Pr_b[\bidsMoreReserve \leq \blocksize] \expect_{\mathbf{\fee}}{Pay^{UPGA_{\reserve}} - \hspace{-0.5cm}\sum_{\fee_i \in \alloc^\auction(\fee_i, \mathbf{\fee_{-i}})} \hspace{-0.5cm} \burn^\auction_i(\fee_i, \mathbf{\fee_{-i}}) | \bidsMoreReserve \leq \blocksize}
        ,
    \end{align*}
    and
    \begin{align*}
         &
         Pr_b[\bidsMoreReserve \leq \blocksize] \expect_{\mathbf{\fee}}{Pay^{UPGA_{\reserve}} - \hspace{-0.5cm}\sum_{\fee_i \in \alloc^\auction(\fee_i, \mathbf{\fee_{-i}})} \hspace{-0.5cm} \burn^\auction_i(\fee_i, \mathbf{\fee_{-i}}) | \bidsMoreReserve \leq \blocksize}
         \\&
         =
         Pr_b[\bidsMoreReserve \leq \blocksize] \expect_{\mathbf{\fee}}{\reserve - \reserve | \bidsMoreReserve \leq \blocksize}
         =
         0
         .
    \end{align*}
\end{proof}

Our characterization of \gls{OCA}-proof \glspl{TFM} (see \cref{thm:OcaChar}) equips us with the tools to understand which mechanism may outperform the \gls{PABGA}, once we step out of the class of well-reserved mechanisms.
In particular, the optimality of the \gls{PABGA} can be examined through the lens of the ``burn difference'' function $\burn^{\text{\tiny diff}}$, that, once chosen, determines the allocation, and thus, the payments.
\begin{definition}[Burn difference]
    \label{def:BurnDiff}
    Given the burn rule $\burn$ of an \gls{OCA}-proof \gls{TFM}, and the matching size-based burn function $\burn^{\text{\tiny cardinal}}$, then the \emph{burn difference} function is defined as:
    $
        \burn^{\text{\tiny diff}}(k)
        =
        \begin{cases}
            \burn^{\text{\tiny cardinal}}(1)                      & k = 1                    \\
            \burn^{\text{\tiny cardinal}}(k) - \burn^{\text{\tiny cardinal}}(k-1) & 2 \leq k \leq \blocksize
        \end{cases}
        .
    $
\end{definition}
If $\burn^{\text{\tiny diff}}$ is constant, we know that the \gls{PABGA} is optimal.
However, as we show in \cref{ex:NonConstantDelta}, this may not hold in other cases.
Intuitively, an increasing burn difference function allows us to limit the supply, and therefore increase the revenue relative to the \gls{PABGA}, implying that the latter is not optimal.
\begin{example}
    \label{ex:NonConstantDelta}
    Consider $4$ bidders, where each bidder $i$ is distributed \gls{iid} $v_i \sim UNI([0,1])$.
    Given these bidders, the expected revenue of the \gls{PABGA} on $3$ items is $\frac{3}{5}$, while the expected revenue of the \gls{PABGA} on $2$ items is $\frac{4}{5}$ (this follows from our calculations in \cref{clm:UniformDistB}).
    We conclude that the mechanism that first limits supply and then runs the \gls{PABGA} outperforms the ``vanilla'' \gls{PABGA}.
    In terms of $\burn^{\text{\tiny diff}}$, this means that instead of having $\burn^{\text{\tiny diff}}(1) = \burn^{\text{\tiny diff}}(2) = \burn^{\text{\tiny diff}}(3) = 0$, which yields \gls{PABGA}, we have $\burn^{\text{\tiny diff}}(1) = \burn^{\text{\tiny diff}}(2) = 0, \burn^{\text{\tiny diff}}(3) = \infty$.
\end{example}

\section{Discussion}
\label{sec:Discussion}
We highlight directions that may cultivate interesting future research on \glspl{TFM}:

\paragraphNoSkip{Complex valuations of transactions}
Previous research on \gls{TFM} focuses on users with unit or additive demand, in the sense that they either have a single transaction to be processed, or they have several that are independent of each other \cite{lavi2019redesigning,roughgarden2021transaction,chung2023foundations,shi2023what,chung2024collusionresilience,gafni2024barriers}.
Realistically, a user's utility from a given transaction may depend on other transactions \cite{ferreira2023credible,yaish2024speculative}.
Nourbakhsh \emph{et al.} \cite{nourbakhsh2024transaction} advance a formalization of this setting and show that \glsxtrshort{EIP}1559 is not \gls{IC} when considering order-sensitive transactions.
Designing mechanisms that ensure good properties under such settings is an open question, more so if they are required to do so efficiently (preferably sub-linearly, perhaps using sketching algorithms).

\paragraphNoSkip{``Optimal'' \glspl{TFM}}
We take the common benchmark of considering auctions to be optimal when the auctioneer maximizes its revenue.
Although it can be argued that miners should be adequately compensated as they invest resources to secure the system\iflong \cite{yaish2023correct}\fi, one may consider other benchmarks for the optimality of \glspl{TFM}.
For example, Basu \emph{et al.} \cite{basu2023stablefees} consider it a feature of their design that users are expected to pay less.
It could be interesting to consider revenue-\emph{minimizing} auctions that maintain high user utility, or study an optimization problem over user utility, subject to maximal revenue constraints.

\paragraphNoSkip{Discrete \glspl{TFM}}
While our impossibility result for discrete-bid \glspl{TFM} shows that their revenue is limited by the number of transactions included in a block, our analysis focuses on the case where each transaction pays $1$, thus leading to low revenue.
A promising direction for research is to consider \glspl{TFM} that set the discrete ``steps'' for bids in a manner that maximizes revenue, similarly to the literature on discrete-step English auctions \cite{david2007optimal}.

\paragraphNoSkip{Plain model \glspl{TFM}}
A line of work shows that various auction designs can prevent misbehavior by actors, for example, by relying on dynamic rules and fee burning (see \cref{rmk:DynamicTFMs}), or on cryptographic primitives that can be augmented with randomized allocation rules \cite{shi2023what}.
However, such designs introduce complexity that may result in unintended side effects.
In the case of dynamic rules, common designs are prone to ``chaotic'' behavior \cite{leonardos2021dynamical,reijsbergen2021transaction,leonardos2023optimality}, while obtaining randomness in the blockchain setting can be hard \cite{gafni2024barriers}.
In contrast, we find that the \gls{PABGA} is good (it is \gls{MMIC}, \gls{OCA}-proof and has good revenue guarantees) without employing cryptographic tools, thus fitting the so-called plain model considered by Shi \emph{et al.} \cite{shi2023what}.
Furthermore, it is not dynamic, and does not rely on fee burning or randomization.
These findings highlight the need to explore similarly ``simple'' mechanisms \iflong\cite{daskalakis2011simple,devanur2015simple}\else\cite{daskalakis2011simple}\fi, and the trade-offs that may exist between simplicity and the ability to satisfy the \gls{TFM} desiderata.

\paragraphNoSkip{Non-constant $\burn^{\text{\tiny diff}}$ \glspl{TFM}}
In our characterization of \gls{OCA}-proof mechanisms in \cref{thm:OcaChar}, and later in the characterization of the intersection of it with the class of well-reserved \glspl{TFM} in \cref{lem:wellReservedOCA}, we effectively limit ourselves to examining mechanisms with a constant size-based burn difference function $\burn^{\text{\tiny diff}}$.
As we show in \cref{ex:NonConstantDelta}, this is significant.
Non-constant $\burn^{\text{\tiny diff}}$ mechanisms generalize the idea of limiting supply (previously studied by Roughgarden \emph{et al.} \cite{roughgarden2012supply}), and are a promising design for high-revenue \gls{OCA}-proof mechanisms.

\begin{credits}
    \bigskip\noindent\textbf{\ackname}
    The authors' work on a related report \cite{gafni2022greedy} that includes early versions of some of the results contained in this work was supported by the European Research Council under the European Union’s Horizon 2020 research \& innovation programme (Grant 740435), and the Ministry of Science \& Technology, Israel.
    The work was supported by ``The Latest in De-Fi Research'' grant by the Uniswap Foundation.
\end{credits}

\bibliographystyle{splncs04}
\bibliography{main}

\appendix

\newpage

\section{Proofs}
\label[appendix]{sec:Proofs}

\thmGta*
\begin{proof}
    We tackle each desired \gls{TFM} property separately.

    \emph{(\gls{GTA} is \gls{DSIC}).}
    First, examine a user with value $v_i < 1$.
    If the user bids $\fee_i \geq 1$, its transaction will be allocated, resulting in a negative utility.
    Thus, a rational user will prefer to bid $0$.
    On the other hand, if $v_i \geq 1$, then by bidding less than $1$ the user obtains a utility of $0$, as its transaction will not be allocated.
    But, if the bidder were to bid $\fee_i \geq 1$, then for all $\fee_{-i}$ we get:
    $v_i - 1 = u_i(v_i, \mathbf{\fee_{-i}}) \geq u_i(\fee_i, \mathbf{\fee_{-i}}) = v_i - 1$, thereby satisfying \cref{def:Dsic}.

    \emph{(\gls{GTA} is \gls{MMIC}).}
    \gls{GTA}'s payment rule is separable, thus a given bid cannot affect the payment of another.
    Furthermore, as by assumption block space is infinite, and as by design \gls{GTA} does not burn any fees, it is revenue-maximizing for the miner to include all transactions that bid at least $1$, per the intended allocation rule.
    By \cref{fct:SeparableMaximizingMmic}, we conclude that \gls{GTA} is \gls{MMIC}.

    \emph{(\gls{GTA} is \gls{SCP}).}
    As fees are not burnt, the miner's revenue equals the sum of user payments, implying that the joint utility of the miner and any subset of bidders equals the sum of this subset's allocated bidders' valuations.
    Finally, the joint utility of any coalition is maximized when all agents are truthful, because then we get that any user with a valuation that is strictly positive is allocated.

    \emph{(\gls{GTA} is \gls{OCA}-proof).}
    Given that \gls{GTA} is \gls{SCP}, one can reach this result by making use of \cref{fct:ScpImpliesOca} (proven as Claim~\href{https://arxiv.org/pdf/2402.08564v1\#claim.3.3}{3.3} by  Gafni and Yaish \cite{gafni2024barriers}).
    \begin{myclaim}[\gls{SCP} $\Rightarrow$ \gls{OCA}-proofness \cite{gafni2024barriers}]
        \label{fct:ScpImpliesOca}
    \end{myclaim}
    Note that Gafni and Yaish \cite{gafni2024barriers} consider a version of the \gls{OCA}-proofness property that is quantified with a parameter $c$ (similarly to $c$-\gls{SCP}), yet their result holds under our definition of the property.
    To ensure the rigor of our analysis, we provide a proof for our non-parameterized definition of \gls{OCA}-proofness.
    \begin{proof}
        Assume towards contradiction that auction $\auction$ is \gls{SCP}, but not \gls{OCA}-proof.
        Thus, there is some $c^*$ for which a coalition of the miner and $c^*$ users can deviate and increase their joint utility to be higher than that of the intended allocation's winning coalition.
        As the auction is \gls{SCP}, by definition it is $c$-\gls{SCP} for any $c$, including $c^*$, so no coalition of the miner and $c^*$ users can increase their joint utility by deviating from the honest protocol, reaching a contradiction.
    \end{proof}
    For completeness, we prove this result via another argument.
    In our proof that \gls{GTA} is \gls{SCP}, we show that honesty maximizes the joint utility of the miner and the allocated bidders, meaning that \cref{def:JointMaximizing} is satisfied, thus \cref{fct:JointMaximizingOca} (proven as Prop.~\href{https://arxiv.org/pdf/2106.01340v3\#theorem.5.11}{5.11} by Roughgarden \cite{roughgarden2021transaction}) implies that \gls{GTA} is \gls{OCA}-proof.
    \begin{restatable}[Joint utility maximization]{definition}{defJointMaximizing}
        \label{def:JointMaximizing}
        An allocation strategy $\alloc$ and bidding strategy $s$ are called \emph{joint utility maximizing} if they maximize the total utility of the miner and the creators of all allocated transactions under $\alloc$ over any possible outcome for any valuation vector $\mathbf{v}$.
        Thus, if given $\mathbf{v}$ the corresponding bids (including ``fake'' miner-created bids) are $\mathbf{b} = s(\mathbf{v})$, and the set of all ``real'' bids allocated under $\alloc$ is $C_{\alloc} = \{ i \, \lvert \, \alloc_i(\fee_i,\mathbf{\fee_{-i}}) \define 1 \wedge i \in 1, \dots, n \}$, then for all valuations $\mathbf{v}$, bids $\mathbf{b}'$ and coalitions $C$:
        $
            u_{joint}(C_{\alloc}, s(\mathbf{v}); \mathbf{v})
            \geq
            u_{joint}(C, \mathbf{b}'; \mathbf{v})
            .
        $
    \end{restatable}
    \begin{proposition}[Joint utility maximizing $\Leftrightarrow$ \gls{OCA}-proofness \cite{roughgarden2021transaction}]
        \label{fct:JointMaximizingOca}
        A \gls{TFM} with intended allocation rule $\alloc^\auction$, payment rule $\mathbf{\pay^\auction}$, and burning rule $\mathbf{\burn^\auction}$, is \gls{OCA}-proof \gls{iff} there is an \gls{EPIR} bidding strategy that maximizes the joint utility of the miner and the allocated users under $\alloc^\auction$ over any possible outcome for any valuation vector $v$.
    \end{proposition}

    \emph{(\gls{GTA} obtains strictly positive miner revenue).}
    By definition, the block-space is infinite, and \gls{GTA}'s burning rule is the all $0$ function.
    Furthermore, by the previous results, users bid truthfully and there is no user-miner collusion.
    Thus, for each bidder who has a private value that is larger than $0$, then by the previous results it will submit a transaction with a non-zero fee, and a revenue-maximizing miner will include its transaction and collect a revenue of $1$.
\end{proof}

\SeparableMaximizingMmic*
\begin{proof}
    While our statement of the result differs from Roughgarden's \cite{roughgarden2021transaction} Theorem~\href{https://arxiv.org/pdf/2106.01340v3\#theorem.5.2}{5.2}, note that this is due to differences in how we define payment and burning rules.
    \cref{rmk:SeparableMaximizingMmic} bridges the gap between our works and shows that Roughgarden's proof holds for our statement, when considering our definitions.

    \begin{myremark}[Separable burning rules]
        \label{rmk:SeparableMaximizingMmic}
        Due to subtle modeling differences between \cite{roughgarden2021transaction} and this work, the proof of the latter does not immediately transfer.
        In particular, we consider burning rules that burn fees ``after'' they are paid (i.e., the burn is deducted from the payment), while Roughgarden considers rules that burn fees ``at the source'' (i.e., are deducted from bidders' utilities) and payments that are always transferred in full to the miner.
        This implies that the revenue under our setting (i.e., the payment minus the burn) is equivalent to the payment under Roughgarden's setting, allowing to re-use the same proof for our case.
        Thus, although the practical meaning of both results is identical, Roughgarden's statement is for \emph{any} burn rule, while our statement is only for \emph{separable} ones.
    \end{myremark}
\end{proof}

\thmScpImpliesLowRev*
\begin{proof}
    Let $\blocksize$ be the (potentially infinite) blocksize, $\burn^{\text{\tiny cardinal}}$ be the size-based burn, and $\burn^{\text{\tiny diff}}$ be the burn difference function (see \cref{def:BurnDiff}).

    First, assume towards contradiction that the block size is finite.
    Consider the set of $\blocksize + 1$ bids $\burn^{\text{\tiny cardinal}}(\blocksize) + 100, \ldots, \burn^{\text{\tiny cardinal}}(\blocksize) + 100, 0$.
    By \cref{thm:OcaChar}, all $\blocksize$ non-zero bids are allocated.
    Due to the critical bid payment property of \gls{DSIC} auctions \cite{myerson1981optimal}, each bidder needs to pay the critical bid to be allocated, keeping all other bids fixed.
    In the discrete case, this is either $\burn^{\text{\tiny diff}}(\blocksize)$ or $\burn^{\text{\tiny diff}}(\blocksize) + 1$ (depending on tie-breaking in the choice of argmax of the allocation rule of \cref{thm:OcaChar}).
    Since $\burn^{\text{\tiny cardinal}}$ is monotonically non-decreasing, $\burn^{\text{\tiny diff}}$ is non-negative, and
    $
        \burn^{\text{\tiny cardinal}}(\blocksize)
        =
        \sum_{k=1}^{\blocksize} \burn^{\text{\tiny diff}}(k)
        \geq
        \burn^{\text{\tiny diff}}(\blocksize)
        .
    $
    Thus, the miner and the $0$-bid bidder can collude and increase their joint utility in violation of $1$-\gls{SCP}:
    if the bidder bids above the critical bid and less than the rest (e.g., $\burn^{\text{\tiny diff}}(\blocksize) + 50$), the burn is unchanged (per \cref{thm:OcaChar}, the burn is size-based), but the winners' payments must increase.
    We conclude that the blocksize is infinite.

    With an infinite blocksize, assume toward contradiction that there is some $k$ where $\burn^{\text{\tiny diff}}(k) \geq \burn^{\text{\tiny diff}}(1) + 3$, and take the minimal such $k$.
    Consider the $k$ sets of bids
    $
        \forall j \in \left[k\right]:
        \,
        B_j
        \define
        \{ \burn^{\text{\tiny cardinal}}(k) + 100, \ldots, \burn^{\text{\tiny cardinal}}(k) + 100 \}
        \, s.t. \, |B_j| = j
        .
    $
    Similarly to before, the payment $p_1$ of the single bid in $B_1$ is either $\burn^{\text{\tiny diff}}(1)$ or $\burn^{\text{\tiny diff}}(1)+1$.
    The payment $p_k$ of one of the bids in $B_k$ satisfies in both tie-breaking cases: $p_k \geq \burn^{\text{\tiny diff}}(k) \geq \burn^{\text{\tiny diff}}(1) + 3 \geq p_1 + 2$.
    Thus, there must be some $j<k$ where the payment $p_{j+1}$ of a bidder in $B_{j+1}$ is larger than the payment $p_j$ of a bidder in $B_j$: $p_{j+1} > p_j$.
    Let $j^*$ be the \textit{last} such index, and consider the set of $j^* + 1$ bids
    $
        B_{j^*}'
        \define
        \{
        \burn^{\text{\tiny cardinal}}(k) + 100, \ldots, \burn^{\text{\tiny cardinal}}(k) + 100, p_{j^*+1} - 1
        \}
        .
    $
    Fixing the first $j^*$ bidders to $\burn^{\text{\tiny cardinal}}(k) + 100$, then to be allocated one needs to bid above the critical bid $p_{j^*+1}$, implying that the $p_{j^*+1} - 1$ bid is unallocated.
    We consider the two tie-breaking cases.
    If the payment of a winning bidder is $p_{j^*} + 1$, then given bids $B_{j^*}$ which are augmented with a $0$-bid bidder, there is a profitable collusion between the miner and the $0$-bid bidder where the bidder increase its bid to $p_{j^*+1} - 1$, as this increases winner payments and thus the coalitions' joint utility (i.e., the miner's revenue).
    Otherwise, if the payment of a winner under $B_{j^*}'$ is $p_{j^*}$, there is a profitable collusion of the miner and the $p^{j^*+1} - 1$-bid bidder:
    the latter increases its bid to $\burn^{\text{\tiny cardinal}}(k) + 100$, is allocated, the payments of the other $j^*$ bidders increase by at least $1$ to $p^{j^*+1}$, and the size-based burn increases by $\burn^{\text{\tiny diff}}(k+1)$.
    By our characterization (\cref{thm:OcaChar}), the critical bid would not be allocated if lower than the burn difference, thus: $\burn^{\text{\tiny diff}}(k+1) \leq p^{j^*+1}$.
    Overall, the change in the miner and bidder's joint utility is $j^* - 1$.
    Thus, as long as $j^* \geq 2$, this contradicts $1$-\gls{SCP}.
    If $j^* = 1$, then $\burn^{\text{\tiny diff}}(2) \geq \burn^{\text{\tiny diff}}(1) + 3$, and so $p_2 \geq p_1 + 2$, implying the miner can coordinate with a $0$-bid bidder given $B_1$ to raise its bid to $p_2 - 1$, which is still unallocated but increases the payment for the $\burn^{\text{\tiny cardinal}}(k) + 100 = \burn^{\text{\tiny cardinal}}(2) + 100$ bidder.

    Now, assume toward contradiction there is some $k$ with $\burn^{\text{\tiny diff}}(k) \leq \burn^{\text{\tiny diff}}(1) - 3$, and take the first such $k$.
    Consider the payment $p_k$ as previously defined.
    It must be that $p_k \leq \burn^{\text{\tiny diff}}(1) - 3 + 1 = \burn^{\text{\tiny diff}}(1) - 2$.
    Since $k$ is the first index with $\burn^{\text{\tiny diff}}(k) \leq \burn^{\text{\tiny diff}}(1) - 3$, the payment is (weakly) lower than all previous values of $\burn^{\text{\tiny diff}}$ (by assumption, $\forall i < k: \burn^{\text{\tiny diff}}(i) \geq \burn^{\text{\tiny diff}}(1) - 2$), and strictly lower than $\burn^{\text{\tiny diff}}(1)$.
    Therefore, the burn exceeds the payments, contradicting \gls{EPBB}:
    $
        k \cdot p_k
        \leq
        k (\burn^{\text{\tiny diff}}(1) - 2)
        <
        \sum_{i=1}^k \burn^{\text{\tiny diff}}(i)
        =
        \burn^{\text{\tiny cardinal}}(k)
        .
    $

    Overall, we conclude that for any $k$:
    $\burn^{\text{\tiny diff}}(1) - 2 \leq \burn^{\text{\tiny diff}}(k) \leq \burn^{\text{\tiny diff}}(1) + 2$.
    Thus, the payment by any allocated bid is at most $\burn^{\text{\tiny diff}}(1) + 3$.
    Finally, given $k$ allocated bidders, the miner's revenue is bounded by:
    \begin{align*}
        &
        \sum_{i=1}^k \pay_i(\fee_i, \mathbf{\fee_{-i}}) - \sum_{i=1}^k \burn_i(\fee_i, \mathbf{\fee_{-i}})
        \le
        \sum_{i=1}^k (\burn^{\text{\tiny diff}}(1) + 3) - \burn^{\text{\tiny cardinal}}(k)
        \\&
        =
        \sum_{i=1}^k \left(\burn^{\text{\tiny diff}}(1) + 3 - \burn^{\text{\tiny diff}}(i) \right)
        \leq
        \sum_{i=1}^k \left(\burn^{\text{\tiny diff}}(1) + 3 - (\burn^{\text{\tiny diff}}(1) - 2) \right)
        =
        5k 
        \in
        O(k)
    \end{align*}
\end{proof}

\thmSeparableScpImpliesLowRev*
\begin{proof}
Let $\auction$ be an \gls{EPIR} and \gls{EPBB} auction with separable burning and payment rules, specifically denote its payment rule by $\mathbf{p^\auction}$.
Choose an arbitrary order over the bidders, and let $\mathbf{\fee}$ be the maximal revenue sequence of bids such that the miner's revenue exceeds the number of allocated bids.
In case there are again multiple sequences satisfying this condition, pick the one with the minimal number of allocated bids.

Now, assume by contradiction that the statement of \cref{thm:SeparableScpImpliesLowRev} does not hold.
This implies that there is at least one agent $i$ with allocated bid $\fee_i$ that results in the miner earning more than $1$.
Let $\rho = \pay^\auction_i\left(\fee_i, \mathbf{\fee_{-i}}\right) - \burn^\auction_i\left(\fee_i, \mathbf{\fee_{-i}}\right) \geq 2 > 1$ be the miner's revenue from bidder $i$, and denote that bidder's payment by $\sigma = \pay^\auction_i\left(\fee_i, \mathbf{\fee_{-i}}\right)$. 
By \gls{EPIR}, we know that $2 \leq \rho \leq \sigma = \pay^\auction_i\left(\fee_i, \mathbf{\fee_{-i}}\right) \le \fee_i$.

Following the characterization of Myerson \cite{myerson1981optimal}, deterministic \gls{DSIC} auctions have monotone allocations and the critical bid property: a bidder is the critical bid so that it is allocated the item if and only if its bid is higher than this price.
Thus, we know that $\alloc_i(\fee_i, \mathbf{\fee_{-i}}) = 1[\fee_i \geq \sigma]$. 
Since $\sigma \rho \geq 2$, we can consider the case where $\fee'_i = \sigma - 1$.
Then, consider the collusion between the miner and bidder $i$ where the bidder bids $\fee_i$, and the miner pays $1$ to the bidder afterwards.
Thus, the bidder's bid is allocated and results in a revenue of $\rho \geq 2$ to the miner, allowing the bidder to be reimbursed.
The joint utility of the colluding miner and user increases from $0$ to $u_i(\fee_i, \mathbf{\fee_{-i}} ; \sigma - 1) + \rho = (\sigma - 1 - \sigma) + \rho = \rho - 1 \geq 1$. 

Recall that due to our assumption, specifically, that the payment and burning rules are separable, we have that the revenue from other bids is unchanged.
\end{proof}

\lemShadingAuction*
\begin{proof}
    We examine the different properties separately.

    \emph{The auction is \gls{BIC}.}
    Chawla and Hartline \cite{chawla2013auctions} analyze a class of auctions that includes the \gls{PABGA}, and show in their Theorem 4.6 (restated in \cref{thm:UniquenessOfBne}) that any auction from this class has a unique \gls{BNE} if bidders are \gls{iid} and their valuations are continuous and have bounded support.
    Our setting corresponds to the one of \cite{chawla2013auctions} (see \cref{sec:Model}), thus we conclude by applying \cref{fct:RevelationPrinciple}.
    \begin{theorem}[Uniqueness of \gls{BNE} \cite{chawla2013auctions}]
        \label{thm:UniquenessOfBne}
        There is only one symmetric Bayes-Nash Equilibrium of an n-agent auction with win-vs-tie-strict rank-based allocation rule with a (potentially 0) reserve price, bid-based payment rule, and \gls{iid} continuous distribution on values.
    \end{theorem}

    \emph{(The auction is \gls{MMIC}).}
    Under the \gls{BNE} considered when showing that the auction is \gls{BIC}, the auction is separable (\cref{def:SeparableRules}), as the payment of each winning bid $\fee_i$ does not depend on the other winning bids, but rather only on $\fee_i$ itself.
    The allocation prioritizes bids that result in higher payments, and thus is revenue maximizing, where revenue maximization is conditioned on the payment rule and does not mean that the auction is revenue optimal among all possible auction formats.
    By \cref{fct:SeparableMaximizingMmic}, we conclude that the auction is \gls{MMIC}.

    \emph{(The auction is efficient).}
    The unique \gls{BNE} shown by Chawla and Hartline \cite{chawla2013auctions} is \emph{efficient}, i.e., the bidders with the highest valuations are allocated, where this is implied by the ``rank-based'' allocation rule.

    \emph{(The auction is \gls{OCA}-proof).}
    Since the shading auction is efficient and does not burn fees, the joint utility of the miner and the allocated users is maximized and the conditions of \cref{fct:JointMaximizingOca} are satisfied, implying \gls{OCA}-proofness.
\end{proof}

\BKPABGA*
\begin{proof}
    The proof is a corollary of the following claims and the revenue equivalence between the \gls{PABGA} and the \gls{UPGA}:

    \begin{myclaim}
        The optimal revenue auction with $n - \blocksize$ bidders is a $\frac{n - \blocksize}{n}$ approximation to the optimal revenue auction with $n$ bidders.
    \end{myclaim}
    \begin{proof}
        By Myerson's theorem \cite{myerson1981optimal}, the revenue of the optimal auction $\auction_n$ with $n$ bidders is:
        $
            \argmax_{x} \expect_{v \sim F^n}{\sum_{i=1}^n \phi(v_i)x_i(v)}
            =
            n \argmax_{x_i} \expect_{v_i \sim F^n}{\phi(v_i)x_i(v)}
            ,
        $
        where $F,f$ are respectively the \gls{CDF} and \gls{PDF} of the distribution, and $\phi(v_i) = v_i - \frac{1 - F(v_i)}{f(v_i)}$.
        Since $x_i$ with $n-k$ bidders dominates $x_i$ with $n$ bidders (as there are less competing bidders for the same items),

        \begin{align*}
            Rev^{\auction_{n-\blocksize}}
            & \geq
            (n - \blocksize)  \argmax_{x_i} \expect_{v_i \sim F^{n-\blocksize}}{\phi(v_i)x_i(v)}
            \nonumber\\
            & \geq
            \frac{n - \blocksize}{n} \cdot n \argmax_{x_i} \expect_{v \sim F^{n-\blocksize}}{\phi(v_i)x_i(v)}
            \nonumber\\
            & \geq
            \frac{n - \blocksize}{n} \cdot n \argmax_{x_i} \expect_{v_i \sim F^{n}}{\phi(v_i)x_i(v)}
            \nonumber\\
            & =
            \frac{n - \blocksize}{n} Rev^{A_n}
        \end{align*}
    \end{proof}

    \begin{myclaim}[Bulow-Klemperer for $\blocksize$ identical items]
        The \gls{UPGA} with $n$ bidders has at least the revenue of the optimal auction with $n - \blocksize$ bidders.
    \end{myclaim}

    \begin{proof}
        The \gls{UPGA} with $n$ bidders is the optimal revenue auction that always allocates all $\blocksize$ items.
        Given the primitive of the optimal revenue auction $\auction$ with $n - \blocksize$ bidders, we can design an auction with the same revenue as $\auction$ that always allocates all $\blocksize$ items, by giving all the unallocated items for free to (perhaps some) of the remaining $\blocksize$ bidders.
        But, this auction then cannot supersede the revenue of the \gls{UPGA} as it is in the class for which the \gls{UPGA} is optimal.
    \end{proof}
\end{proof}

\uniformDist*

\begin{proof}
    By revenue equivalence and direct calculation:
    \begin{align}
        \expect_{v_j \sim UNI([0,1])}{Rev^{PABGA}}
         & =
        \expect_{v_j \sim UNI([0,1])}{Rev^{UPGA}}
        \nonumber \\
         & =
        \blocksize \cdot \expect_{v_j \sim UNI([0,1])}{\left( \{v_j\}_{j=1}^n \right)_{\left(n-(\blocksize+1)\right)}}
        \nonumber \\
         & =
        \blocksize \cdot \frac{n - \blocksize}{n + 1}
        \nonumber
    \end{align}
    The last transition is a calculation of the $\left(n-(\blocksize + 1)\right)$-th order statistic, i.e., the $\blocksize+1$ highest bid (in our case, the top rejected bid).
    Furthermore:
    \begin{align}
        \expect_{v_j \sim UNI([0,1])}{V}
         & =
        \sum_{i=1}^{\blocksize} \expect_{v_j \sim UNI([0,1])}{\left( \{v_j\}_{j=1}^n\right)_{(n-i)}}
        \nonumber \\
         & =
        \sum_{i=1}^{\blocksize} \frac{n + 1 - i}{n + 1}
        \nonumber \\
         & =
        \frac{\blocksize}{n+1}(n + 1 - \frac{\blocksize + 1}{2})
        \nonumber
    \end{align}
\end{proof}

\UniformDistB*

\begin{proof}
    It is enough to show that:
    $
        \forall v_j \in [0,1]: s(v_j) \geq \frac{n - \blocksize}{n} v_j
        .
    $
    Then:
    \begin{align*}
        \expect_{v_j \sim UNI([0,1])}{\frac{Rev^{PABGA}}{V}}
         & =
        \expect_{v_j \sim UNI([0,1])}{\frac{\sum_{i=1}^{\blocksize} s(v_j)}{\sum_{i=1}^{\blocksize}v_j}}
        \\
         & \geq \expect_{v_j \sim UNI([0,1])}{\frac{\sum_{i=1}^{\blocksize} \frac{n - \blocksize}{n} v_j}{\sum_{i=1}^{\blocksize}v_j}}
        \\
         & =
        \frac{n - \blocksize}{n} \expect_{v_j \sim UNI([0,1])}{\frac{\sum_{i=1}^{\blocksize} v_j}{\sum_{i=1}^{\blocksize} v_j}}
        \\
         & = \frac{n - \blocksize}{n}
    \end{align*}

    Thus, it is enough to prove the following claim:
    \begin{myclaim}
        With $n$ bidders \gls{iid} $UNI\left(\left[0,1\right]\right)$ and $\blocksize \leq 10$ items, there is a Bayesian Nash equilibrium where bidders with value $v$ bid $s(v)$, where
        \begin{align*}
            s(v)
             & =
            \frac{(n-k)v}{n} \cdot \frac{P_{n,k}(v)}{P_{n-1,k}(v)}
            \\
            P_{n,k}(v)
             & =
            \sum_{i=0}^{k-1} \binom{n}{n-i} \binom{n-i-1}{n-k} (-v)^{k-1-i}
        \end{align*}
        Moreover, $\frac{P_{n,k}(v)}{P_{n-1, k}(v)} \geq 1$ for any $0 \leq v \leq 1$ and $n \geq \blocksize$.
    \end{myclaim}

    \begin{proof}
        We generally follow the lines of the discussion of the derivation of the first-price auction equilibrium in \cite{easley2010networks}, but generalized to the \gls{PABGA}.
        We use $s(v)$ as specified in the statement as a guess to solve the differential equation:

        \begin{equation*}
            \frac{\partial \sum_{i=0}^{k-1} \binom{n-1}{i} v^{n-1-i}(1-v)^i \left( x - s(v) \right)}{\partial v} (x) = 0
        \end{equation*}

        This can then be directly verified using Mathematica up to $\blocksize = 10$. The fact that $P_{n,\blocksize}(v) \geq P_{n-1, \blocksize}(v)$ can be directly verified as well using multi-variable polynomial positivity methods (in particular, cylindrical algebraic decomposition \iflong\cite{caviness2012quantifier}\fi, previously used in mechanism design contexts by Gafni \emph{et al.} \cite{gafni2020vcg}).
        This makes use of the fact that $P_{n,\blocksize}(v)$ is polynomial not only in $v$ but in $n$ as well. This was  verified to hold up to $\blocksize = 65$.
    \end{proof}
\end{proof}

\ExponentialDist*

\begin{proof}
    Denote the harmonic sum by $H_n = \sum_{i=1}^n \frac{1}{i}$.
    By revenue equivalence and the formula for exponential \gls{iid} order statistics \cite{renyi1953theory}, we get:
    \begin{align*}
        \expect_{v_j \sim EXP(\zeta)}{Rev^{PABGA}}
        &
        =
        \expect_{v_j \sim EXP(\zeta)}{Rev^{UPGA}}
        \nonumber\\&
        =
        \blocksize \cdot \expect_{v_j \sim EXP(\zeta)}{\left(\{v_j\}_{j=1}^n\right)_{\left(n- (\blocksize + 1)\right)}}
        \nonumber\\&
        =
        \blocksize \frac{1}{\zeta} \cdot \sum_{j=\blocksize + 1}^{n} \frac{1}{j}
        =
        \blocksize \frac{1}{\zeta} \cdot (H_n - H_{\blocksize}),
    \end{align*}

    Furthermore:
    $
        \expect_{v_j \sim EXP(\zeta)}{V}
        =
        \sum_{i=1}^{\blocksize}  \expect_{v_j \sim EXP(\zeta)}{\left( \{v_j\}_{j=1}^n \right)_{\left(n-(i+1)\right)}}
        \leq
        \blocksize \frac{1}{\zeta} H_n
        .
    $
    %
    By writing $H_n = \ln n + O(1)$ \cite{boas1971partial}, and since for some $\nu$ and $\blocksize$ large enough we have $n \geq \nu \blocksize^c$, we can write for any $\blocksize$ large enough,
    \begin{align*}
        \frac{\expect_{v_j \sim EXP(\zeta)}{Rev^{PABGA}}}{\expect_{v_j \sim EXP(\zeta)}{V}}
        &
        \geq
        \frac{\blocksize \frac{1}{\zeta} \cdot (H_n - H_{\blocksize})}{\blocksize \frac{1}{\zeta} H_n}
        =
        \frac{H_n - H_{\blocksize}}{H_n}
        \nonumber\\&
        \geq
        \frac{H_{\nu \blocksize^c} - H_{\blocksize}}{H_{\nu \blocksize^c}}
        \geq
        \frac{\ln(\blocksize^c) - \ln(\blocksize) - O(1)}{\ln(\blocksize^c) + O(1)}
        \nonumber\\&
        =
        \frac{(c-1)\ln(\blocksize) - O(1)}{c \ln (\blocksize) + O(1)}
    \end{align*}
    Following from the above, we get:
    \begin{align*}
        lim_{\blocksize \rightarrow \infty} \frac{\expect_{v_j \sim EXP(\zeta)}{Rev^{PABGA}}}{\expect_{v_j \sim EXP(\zeta)}{V}}
        & \geq
        \lim_{\blocksize \rightarrow \infty} \frac{(c-1)\ln(\blocksize) - O(1)}{c \ln (\blocksize) + O(1)}
        =
        \frac{c-1}{c}
    \end{align*}
\end{proof}

\WellReservedOCA*
\begin{proof}
    Let $\blocksize$ be the block size, $\burn^{\text{\tiny cardinal}}$ be the size-based burn function, and $\burn^{\text{\tiny diff}}$ be the burn difference function.
    We now show that $\burn^{\text{\tiny diff}}$ must be a constant function if the mechanism is well-reserved.
    We proceed by splitting into cases according to whether the first divergent value of $\burn^{\text{\tiny diff}}$ is higher or lower than the preceding ones.

    Assume toward contradiction that $\burn^{\text{\tiny diff}}(1) = R$, and that there is some $k$ such that $\forall i < k, \burn^{\text{\tiny diff}}(i) = R$, and $\burn^{\text{\tiny diff}}(k) = \reserve < R$.
    For any $\reserve' > \reserve$, consider the vector of $k$ bids $\frac{\reserve' + \reserve}{2}, R, \ldots, R$.
    By the allocation rule of \cref{thm:OcaChar}, all bids are allocated.
    We conclude that $\reserve'$ cannot be the reserve price of well-reserved \glspl{TFM}, as it must satisfy that no bidder can be allocated when bidding less than $\reserve'$, but $\frac{\reserve'+\reserve}{2}$ is allocated.
    For any $\reserve' \leq \reserve$, consider the single bid $\frac{R + \reserve'}{2}$.
    By the allocation rule of \cref{thm:OcaChar}, the bid is not allocated, as it is less than $R$.
    We conclude that $\reserve'$ cannot be the reserve price of the well-reserved \glspl{TFM}, as it must satisfy that a bid cannot be unallocated if it is above $\reserve'$ and the number of allocated bids does not reach $\blocksize$.
    Since this covers all possible reserve prices for the current case, we conclude that the mechanism is not well-reserved.

    Now, consider the other case and assume, in contradiction, that $\burn^{\text{\tiny diff}}(1) = R$, and there is some $k$ so that $\forall i < k, \burn^{\text{\tiny diff}}(i) = R$, and $\burn^{\text{\tiny diff}}(k) = \reserve > R$.
    For any $\reserve' < \reserve$, consider the vector of $k$ bids $\frac{\reserve' + \reserve}{2}, \reserve, \ldots, \reserve$.
    By \cref{thm:OcaChar}, all bids are allocated except $\frac{\reserve' + \reserve}{2}$.
    Thus, $\reserve'$ cannot be the reserve price of well-reserved \glspl{TFM}, as it must satisfy that a bid cannot be unallocated if above $\reserve'$ and the number of allocated bids does not reach $\blocksize$.
    For $\reserve' \geq \reserve$, consider the bid $\frac{\reserve' + \reserve}{2}$.
    By \cref{thm:OcaChar}, the bid is allocated (since it is larger than $R$), but it should not (as it is less than $\reserve'$).
    This covers all options, thus the mechanism is not well-reserved.

    Finally, notice that given a constant $\burn^{\text{\tiny diff}}$, the allocation rule of \cref{thm:OcaChar} allocates the highest bidders as long as they are at least as high as $\burn^{\text{\tiny diff}}(1)$, while $\burn^{\text{\tiny diff}}(1)$ is burnt for each allocated bidder.
\end{proof}

\iflong
\newpage
\section{Glossary}
\label[appendix]{sec:Glossary}
We now present a summary of the symbols and acronyms used in this work.
\setglossarystyle{alttree}
\glssetwidest{AAAAAAA}
\printnoidxglossary[type=symbols]
\printnoidxglossary[type=\acronymtype]
\fi

\end{document}

%% file: preamble.tex
\newif\iflong
\longtrue

\newif\ifcomments   
\commentsfalse

\newcommand{\paragraphNoSkip}[1]{\par\smallskip\noindent\textbf{#1}.}

\usepackage[T1]{fontenc}    
\usepackage[english]{babel} 
\usepackage{csquotes}       
\usepackage{dsfont}         
\usepackage{mathtools}      
\usepackage{enumitem}       
\usepackage{graphicx}       
\graphicspath{ {./images/} }
\usepackage{amsfonts}       
\usepackage{hyperref} 
\usepackage{xcolor}         
\hypersetup{                
    colorlinks,
    linkcolor={red!50!black},
    citecolor={blue!50!black},
    urlcolor={blue!80!black}
}
\usepackage{etoolbox}\appto\UrlBreaks{\do\-}
\usepackage[capitalise]{cleveref}       

\ifcomments\fi  
\usepackage[colorinlistoftodos,prependcaption,textsize=tiny,textwidth=\marginparwidth]{todonotes}

\usepackage{thm-restate}   
\spnewtheorem{junk}{Junk}{\bfseries}{\itshape} 
\spnewtheorem{myclaim}[theorem]{Claim}{\bfseries}{\itshape}
\spnewtheorem{myremark}[theorem]{Remark}{\bfseries}{\itshape}
\spnewtheorem{fact}[theorem]{Fact}{\bfseries}{\itshape}
\spnewtheorem{conclusion}[theorem]{Conclusion}{\bfseries}{\itshape}

\usepackage{etoolbox} 
\AtEndEnvironment{proof}{\phantom{}\qed} 

\crefname{definition}{Def.}{Defs.}
\crefname{theorem}{Theorem}{Theorems}
\crefname{myclaim}{Claim}{Claims}
\crefname{proposition}{Prop.}{Props.}
\crefname{fact}{Fact}{Fact}
\crefname{lemma}{Lemma}{Lemmas}
\crefname{corollary}{Corollary}{Corollaries}
\crefname{myremark}{Remark}{Remarks}

\DeclareMathOperator*{\argmax}{arg\,\max}
\newcommand{\define}{\stackrel{\mathclap{\mbox{\text{\tiny def}}}}{=}}

\usepackage[symbols,acronym,nonumberlist,nogroupskip,section=subsection,numberedsection,stylemods={mcols,longbooktabs}]{glossaries-extra}
\iflong\makenoidxglossaries\fi

\glssetcategoryattribute{acronym}{nohyper}{true} 
\setabbreviationstyle[acronym]{long-short}  
\newacronym{DEX}{DEX}{decentralized exchange}
\newacronym{EVM}{EVM}{Ethereum virtual machine}
\newacronym{DeFi}{DeFi}{decentralized finance}
\newacronym{PoW}{PoW}{Proof-of-Work}
\newacronym{PoS}{PoS}{Proof-of-Stake}
\newacronym{MEV}{MEV}{miner-extractable value}
\newacronym{block-DAG}{block-DAGs}{block directed-acyclic-graph}
\newacronym{DAA}{DAA}{difficulty-adjustment algorithm}
\newacronym{MDP}{MDP}{Markov decision-process}
\newacronym{DQL}{DQL}{Deep-Q-learning}
\newacronym{RL}{RL}{reinforcement learning}
\newacronym{ML}{ML}{machine learning}
\newacronym{AI}{AIs}{artificial intelligence}
\newacronym{PDF}{PDF}{probability density function}
\newacronym{CDF}{CDF}{cumulative density function}
\newacronym{AMM}{AMM}{automated market maker}
\newacronym{USD}{USD}{United States Dollar}
\newacronym{IP}{IP}{Internet Protocol}
\newacronym{LP}{LP}{Liquidity Provider}
\newacronym{LT}{LT}{Liquidity Taker}
\newacronym{APY}{APY}{annual percentage yield}
\newacronym{PID}{PID}{proportional integral derivative}
\newacronym{UTXO}{UTXO}{unspent transaction output}
\newacronym{YAML}{YAML}{YAML Ain't Markup Language}
\newacronym{TD}{TD}{total difficulty}
\newacronym{geth}{geth}{Go Ethereum}
\newacronym{WETH}{WETH}{Wrapped Ethereum}
\newacronym{ASIC}{ASIC}{Application Specific Integrated Circuit}
\newacronym{RPC}{RPC}{remote procedure call}
\newacronym{RUM}{RUM}{riskless uncle maker}
\newacronym{PUM}{PUM}{preemptive uncle maker}
\newacronym{URL}{URL}{uniform resource locator}
\newacronym{SSD}{SSD}{solid state drive}
\newacronym{EIP}{EIP}{Ethereum improvement proposal}
\newacronym{CPU}{CPU}{central processing unit}
\newacronym{RAM}{RAM}{random-access memory}
\newacronym{mempool}{mempool}{memory pool}
\newacronym{TFM}{TFM}{transaction fee mechanism}
\newacronym{iid}{i.i.d.}{independent and identically distributed}
\newacronym{iff}{iff}{if and only if}
\newacronym{wrt}{w.r.t.}{with regard to}
\newacronym{wlog}{w.l.o.g.}{without loss of generality}
\newacronym{TTL}{TTL}{time to live}
\newacronym{DSIC}{DSIC}{dominant strategy incentive-compatible}
\newacronym{EPIC}{EPIC}{ex-post incentive-compatible}
\newacronym{BIC}{BIC}{Bayesian incentive-compatible}
\newacronym{IC}{IC}{incentive-compatible}
\newacronym{MMIC}{MMIC}{myopic miner incentive-compatible}
\newacronym{FMIC}{FMIC}{farsighted miner incentive-compatible}
\newacronym{OCA}{OCA}{off-chain-agreement}
\newacronym{SCP}{SCP}{side-contract-proof}
\newacronym{PABGA}{PABGA}{pay-as-bid greedy auction}
\newacronym{SPA}{SPA}{second price auction}
\newacronym{UPGA}{UPGA}{uniform-price greedy auction}
\newacronym{BNE}{BNE}{Bayesian-Nash equilibrium}
\newacronym{EPIR}{EPIR}{ex-post individually rational}
\newacronym{EPBB}{EPBB}{ex-post burn balanced}
\newacronym{GTA}{GTA}{good toy auction}
\newacronym{GTFBA}{GTFBA}{good toy finite-blocksize auction}
\newacronym{QoS}{QoS}{quality of service}
\newacronym{MPC}{MPC}{multi-party computation}

\glsxtrnewsymbol[description={
    An auction.
}]{auction}{
    \ensuremath{A}
}
\newcommand{\auction}{{\gls[hyper=false]{auction}}}

\glsxtrnewsymbol[description={
    A transaction.
}]{tx}{
    \ensuremath{\tau}
}

\glsxtrnewsymbol[description={
    Payment rule.
}]{pay}{
    \ensuremath{p}
}
\newcommand{\pay}{{\gls[hyper=false]{pay}}}

\glsxtrnewsymbol[description={
    Burning rule.
}]{burn}{
    \ensuremath{\beta}
}
\newcommand{\burn}{{\gls[hyper=false]{burn}}}

\glsxtrnewsymbol[description={
    Reserve price of an auction.
}]{reserve}{
    \ensuremath{r}
}
\newcommand{\reserve}{{\gls[hyper=false]{reserve}}}

\glsxtrnewsymbol[description={
    Allocation function.
}]{allocation}{
    \ensuremath{x}
}
\newcommand{\alloc}{{\gls[hyper=false]{allocation}}}

\glsxtrnewsymbol[description={
    Transaction fee of some transaction, in tokens.
}]{fee}{
    \ensuremath{b}
}
\newcommand{\fee}{{\gls[hyper=false]{fee}}}

\glsxtrnewsymbol[description={
    Predefined maximal amount of transactions a block can contain.
}]{blocksize}{
    \ensuremath{\mathcal{B}}
}
\newcommand{\blocksize}{{\gls[hyper=false]{blocksize}}}

\glsxtrnewsymbol[description={
    Number of all bids with $b_i \geq \reserve$.
}]{bidsMoreReserve}{
    \ensuremath{S_{\geq \reserve}}
}
\newcommand{\bidsMoreReserve}{{\gls[hyper=false]{bidsMoreReserve}}}

\NewDocumentCommand{\expect}{ e{_} s o >{\SplitArgument{1}{|}}m }{%
  \operatorname{E}
  \IfValueT{#1}{{\!}_{#1}}
  \IfBooleanTF{#2}{
    \expectarg*{\expectvar#4}%
  }{
    \IfNoValueTF{#3}{
      \expectarg{\expectvar#4}%
    }{
      \expectarg[#3]{\expectvar#4}%
    }%
  }%
}
\NewDocumentCommand{\expectvar}{mm}{%
  #1\IfValueT{#2}{\nonscript\;\delimsize\vert\nonscript\;#2}%
}
\DeclarePairedDelimiterX{\expectarg}[1]{[}{]}{#1}